\documentclass[aps,prx,showpacs,twocolumn,reprint,superscriptaddress]{revtex4-2}

\usepackage{qcircuit}
\usepackage[dvips]{graphicx}
\usepackage{amsmath,amssymb,amsthm,mathrsfs,amsfonts,dsfont,mathtools,physics}
\usepackage{epsfig}
\usepackage{braket}
\usepackage{hyperref}
\usepackage{bm}
\usepackage{enumerate}
\usepackage{color}
\usepackage{graphicx}

\usepackage{enumitem}
\usepackage[capitalize]{cleveref}
\usepackage[normalem]{ulem}
\usepackage{comment}

\usepackage{amsthm}

\newtheorem{definition}{Definition}

\newtheorem{theorem}{Theorem}
\newtheorem{lemma}{Lemma}
\newtheorem{corollary}{Corollary}

\Crefname{theorem}{Theorem}{Theorems}
\theoremstyle{remark}
\newtheorem{algorithm}{Algorithm}

\newcommand{\sign}{\text{sign}}
\renewcommand{\tr}{\mathrm{Tr}}

\renewcommand{\var}{\mathrm{Var}}

\newcommand{\nshot}{N_s}
\newcommand{\nes}{N_E}
\newcommand{\emshadow}{H}

\newcommand{\rhoid}{\rho_{id}}

\newcommand{\qmaddress}{\affiliation{Quantum Motion, 9 Sterling Way, London N7 9HJ, United Kingdom}}
\newcommand{\oxddress}{\affiliation{Department of Materials, University of Oxford, Parks Road, Oxford OX1 3PH, United Kingdom}}

\begin{document}

\title{Quantum Error Mitigated Classical Shadows}
\author{Hamza Jnane}
\email[equal author contributions]{}
\oxddress
\qmaddress

\author{Jonathan Steinberg}
\email[equal author contributions]{}
\affiliation{Naturwissenschaftlich–Technische Fakultät, Universität Siegen, 57068 Siegen, Germany}
\affiliation{State Key Laboratory for Mesoscopic Physics, School of Physics and Frontiers Science Center for Nano-Optoelectronics, Peking University, Beijing 100871, China}

\author{Zhenyu Cai}
\oxddress
\qmaddress

\author{H. Chau Nguyen}
\affiliation{Naturwissenschaftlich–Technische Fakultät, Universität Siegen, 57068 Siegen, Germany}

\author{B\'alint Koczor}
\email{balint.koczor@materials.ox.ac.uk}
\oxddress
\qmaddress

\begin{abstract}
Classical shadows enable us to learn many properties of a quantum state $\rho$ with very few measurements.
However, near-term and early fault-tolerant quantum computers will only be able to prepare noisy quantum states $\rho$
and it is thus a considerable challenge to efficiently learn properties of an ideal, noise free state $\rhoid$. 
We consider error mitigation techniques, such as Probabilistic Error Cancellation (PEC),
Zero Noise Extrapolation (ZNE) and Symmetry Verification (SV) which have been developed
for mitigating errors in single expected value measurements
and generalise them for mitigating errors in classical shadows.
We find that PEC is the most natural candidate and thus
develop a thorough theoretical framework for PEC shadows with the following rigorous theoretical guarantees:
PEC shadows are an unbiased estimator for the ideal quantum state $\rhoid$;
the sample complexity for simultaneously predicting many linear properties of $\rhoid$
is identical to that of the conventional shadows approach up to a multiplicative factor
which is the sample overhead due to error mitigation. Due to efficient post-processing of
shadows, this overhead does not depend directly
on the number of qubits but rather grows exponentially with the number of noisy gates.
The broad set of tools introduced in this work may be instrumental in exploiting
near-term and early fault-tolerant quantum computers: We demonstrate in detailed numerical simulations 
a range of practical applications of quantum computers that will significantly benefit from our techniques.
\end{abstract}

\maketitle
	
\section{Introduction}

Quantum computers are developing rapidly and can already be said to perform
certain demonstration tasks that are impossible or very difficult with even the largest
supercomputers~\cite{aruteQuantumSupremacyUsing2019,zhongPhaseProgrammableGaussianBoson2021,wuStrongQuantumComputational2021,ebadiQuantumPhasesMatter2021,gongQuantumWalksProgrammable2021a}. It is however still to be seen whether the technology
can achieve true practical quantum advantage, i.e., the point when these
machines
can solve an otherwise impossible computational task that is of value to industry
or to researchers in other fields such as quantum field theory~\cite{latticeschwinger},
quantum gravity~\cite{jafferis2022traversable}
or drug development and materials science~\cite{cao2019quantum, mcardle2020quantum,bauer2020quantum, Motta2022}.

Quantum computers are highly vulnerable to noise and while quantum error correction provides
a comprehensive solution, implementing it poses an extreme engineering challenge~\cite{nisq_preskill_2018}. 
It is generally expected that some form of early practical quantum advantage just beyond
the reach of classical computing could be achieved even with noisy quantum computers~\cite{cai2022quantum,kandala2019error,kim2023scalable, o2022purification}.
This prospect has motivated the development of a broad range of quantum error mitigation protocols
which has grown into an entire subfield. While the range of error mitigation tricks are very diverse, 
they collectively aim to mitigate the effect of gate errors in an expected-value measurement process
-- a key subroutine in quantum computing.

\begin{figure*}
	\centering
        \includegraphics[width=\textwidth]{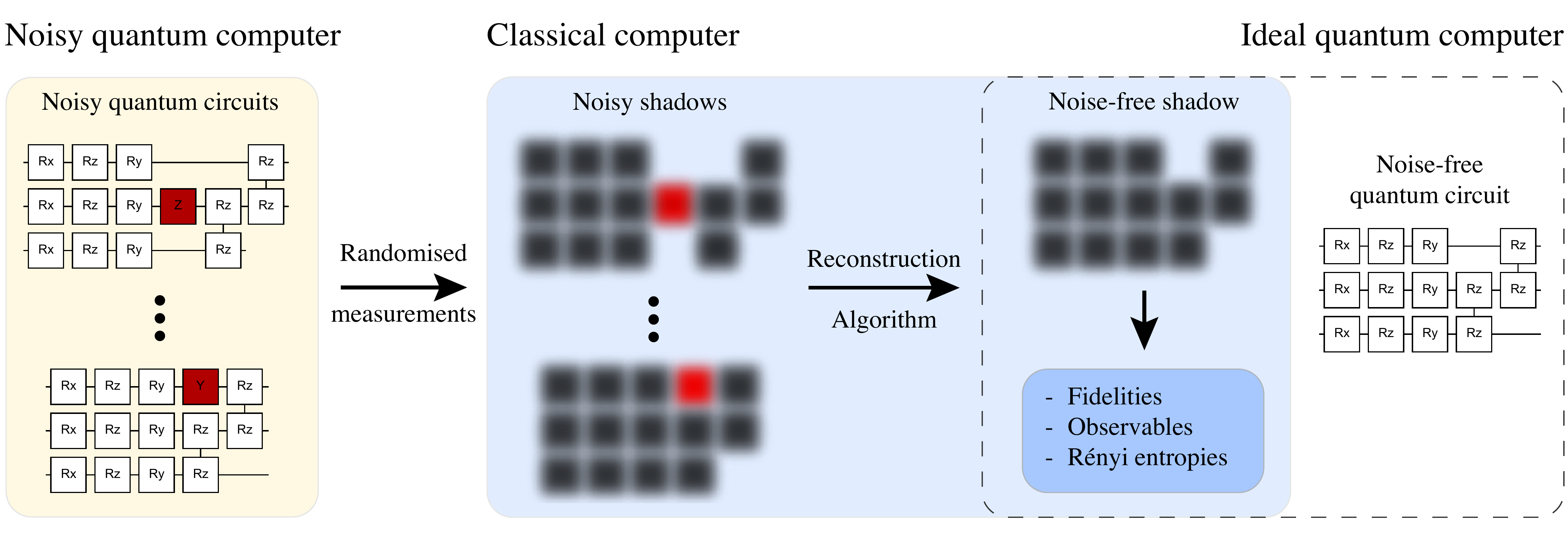}
	\caption{
        In the present work we assume we only have access to a noisy quantum computer (left)
        such that every circuit we run (left, yellow area) gets corrupted by gate noise (unwanted red gate elements).
        We aim to extract properties of a state that would be prepared by an ideal quantum computer (right)
        with the use of powerful error mitigation techniques.
        We provide a rigorous theoretical framework for PEC shadows
        which effectively allows us to obtain a classical shadow of
        the ideal quantum state (noise-free shadow) from which we can predict many ideal
        properties in classical post-processing (middle blue area, classical computer).
        In our formalism we run a series of distinct quantum circuit variants (left, yellow area)
        that cast different classical shadows (noisy shadows) due to gate noise and due to our
        intentional recovery operations (red gate elements and their shadows). 
        Under the assumption that the device's error characteristics have been appropriately learned,
        we can estimate the noise free shadow (middle) via classical post-processing.
        }
	\label{fig:front_fig}
\end{figure*}

Another major challenge is that near-term quantum algorithms typically require an extreme number
of circuit repetitions in order to suppress shot noise~\cite{cerezoVariationalQuantumAlgorithms2021a, endoHybridQuantumClassicalAlgorithms2021, bharti2021noisy,van2021measurement}.
Classical shadows were introduced relatively recently~\cite{shadow_huang_2020}
and represent another promising angle in achieving practical quantum advantage.
The approach allows one to extract many properties of a quantum state without
having to repeat the measurement many times. This is achieved by performing measurements
in randomized bases. The measurement outcomes as bitstrings, along with the indexes of the measurement bases
form a classical shadow which is an efficient classical representation of the entire quantum state.
Shadows have become an entire subfield and various promising applications have been proposed~\cite{covar, shadow_spec_2023} that
greatly benefit from the rich information one can access via shadows. For instance, in 
shadow spectroscopy~\cite{shadow_spec_2023}, we estimate many time-dependent expected values from time-evolved quantum states,
which then allows us to reveal accurate spectra through the use of efficient classical post-processing.

The focus of the present work is to amalgamate quantum error mitigation techniques with 
classical shadows. It is worth noting that prior works have considered fruitful connections between
quantum error mitigation and classical shadows. First, ref.~\cite{seif_distillation_2023, zhou_2023}
use classical shadows obtained from a noisy quantum state to
perform purification-based error mitigation~\cite{PhysRevX.11.031057,huggins2020virtual,PhysRevApplied.18.044064} offline, 
with only access to a single copy of the state but at an exponential complexity in the number of qubits.
Second, the mitigation of errors in the randomised measurements have similarly been
addressed in \cite{chen_2021, koh_classical_2022}.

In the present work we address a distinct problem:
Previous methods have assumed that the task involves extracting information
from a predetermined quantum state $\rho$, such as the output of a quantum device.
However, we consider the practically more relevant scenario where the state $\rho$
is generated by a noisy quantum circuit,
and our aim is to mitigate the impact of errors induced by the noisy quantum gates.
Our focus is thus to extract properties of an ideal state $\rhoid$,
which would be generated by a noise-free quantum computer.
This approach is a generalisation of quantum error mitigation techniques
which generally aim to extract an ideal expected value $\tr[ O \rhoid ]$
when only noisy measurements $\tr[ O \rho ]$ are available.
In contrast, the techniques we present are not restricted to a single
expected value but instead provide efficient classical representations of the
ideal quantum state $\rhoid$ through powerful classical shadows as illustrated in \cref{fig:front_fig}.

While we cover most classes of conventional error mitigation techniques, such as
Probabilistic Error Cancellation (PEC), Zero Noise Extrapolation (ZNE) and
Symmetry Verification (SV),
we find that PEC is the most amiable to be used in combination with classical shadows.
We thus dedicate most attention to PEC shadows for which we establish a comprehensive theory:
Assuming the error model of the quantum gates has been appropriately learned, we rigorously prove
that our PEC shadow is an unbiased estimator of the ideal quantum state $\rhoid$.
We also furnish explicit, efficient classical reconstruction
algorithms that enable the simultaneous prediction of linear and non-linear properties of $\rhoid$.

Similarly to conventional error mitigation techniques, the ability of estimating properties of the noise-free scenario
comes at the cost of an  increased statistical variance which implies an increased number of circuit repetitions.
We establish rigorous bounds on sample complexities and our results indicate that:
(a) the sample complexity of PEC shadows is identical to that of conventional shadows up
to a multiplicative factor, and
(b) this multiplicative factor $\lVert g \rVert_1$ has already been known in the literature as the sample overhead of the
PEC approach~\cite{cai2022quantum}.
Thus the present techniques are efficient in the sense that sample complexities are independent
of the number of qubits -- but of course the overhead grows exponentially with the number of noisy gates.

In numerical simulations we showcase
a broad range of useful practical applications that will play a crucial role in both 
the near-term and in the early fault-tolerance era. These examples comprise:
(a) determining error mitigated energies in variational quantum circuits, which constitutes a fundamental subroutine in near-term applications;
(b) predicting many properties simultaneously in ground state preparation to extract two-point
correlators or to accelerate the training of circuit parameters\cite{shadow_huang_2020, covar, shadow_spec_2023};
(c) extracting error mitigated local entanglement entropies of a ground state that is prepared by a noisy quantum circuit.
Moreover, we discuss several other applications that will significantly benefit from our efficient
amalgams of quantum error mitigation and classical shadows.

This paper is organised as follows. In the following section we first briefly review the formalism
of classical shadows. Then in \cref{sec:pec} we introduce our main result as Probabilistic Error Cancelled
shadows. In \cref{sec:other_mitigation} we discuss how to combine further error mitigation
techniques with classical shadows but conclude that PEC shadows admit the most natural formalism.
Finally, in \cref{sec:app} we demonstrate powerful applications of our approach and then conclude in \cref{sec:concl}.


\section{Preliminaries: classical shadows \label{sec:widealised_measurements}}

The original idea of classical shadow tomography is to apply to the quantum system of $N$ qubits prepared in a specific state $\rho$  a unitary 
$Q_j$ randomly sampled from a certain ensemble $\mathcal{Q}$; typically the ensemble corresponds to just 
rotating the individual qubits with single-qubit unitaries (Pauli basis measurements) or applying Clifford rotations. 
This is followed by a measurement in the computational basis, yielding a bitstring $b \in \{0,1\}^N$ as the outcome; 
This bitstring is logged
along with the measurement basis forming the index $l=(j,b)$.
The collection of these indexes from many independent runs of the protocol then allow us to construct a classical shadow of the state.
A classical shadow provides a description of the quantum
state that can be classically efficiently stored and manipulated, bypassing the computationally-expensive reconstruction of the full density matrix~\cite{shadow_huang_2020}.

\subsection{Classical shadows via idealised measurements}
We mathematically describe a particular measurement outcome $l=(j,b)$ by the positive operator as
$E_{l} {=}  p_{j} Q_j^{\dagger} \vert b \rangle \langle b \vert  Q_j$;
The probability $q_l = \tr (\rho E_l)$  of this outcome is a product of a (classical)
probability $p_j$ of choosing a unitary $Q_j$ and the probability of observing the bitstring $b$
given the rotated measurement basis.
The shadow protocol can be therefore compactly described by a set $E$ of $N_E = 2^N \vert \mathcal{Q} \vert$ positive operators given by
\begin{equation}
\label{eq:idealised_measurements}
E = \{ E_{l} {=}  p_{j} Q_j^{\dagger} \vert b \rangle \langle b \vert  Q_j,
	\, 
	\text{with}
	\,
	Q_j \in \mathcal{Q},
	\,
	b \in \{0,1\}^N
	\}.
\end{equation}
In the literature, such a collection $E$ of positive operators $E_l$ that
sums up to the identity is referred to as a generalised measurement
(positive operator-valued measure - POVM) and $E_l$ are called
its effects~\cite{nielsen_chuang_2010}.
It has been shown that formulating shadow tomography using POVMs brings various advantages~\cite{povm_shadow_nguyen_2022}.
Particularly relevant to our purpose, this formulation allows one to automatically account for errors in measurements, which include both read-out errors and gate errors in the implementation of the random unitaries $Q_j$~\cite{nisq_preskill_2018,detector_tomoIBM_2019,arute_quantum_2019, koh_classical_2022, chen_2021}. 
This is carried out by simply adjusting the effects $E_l$ appropriately~\cite{povm_shadow_nguyen_2022,glos2022adaptive}
as we detail towards the end of this section.

Given the above generalised measurement, a single outcome $l=(j,b)$ can be used to construct a snapshot as
$\hat{\rho}_{l} = C_{E}^{-1} (E_{l})$
where the channel $C_E$ is defined by 
\begin{equation} \label{eq:measurement_channel}
C_{E}(\rho) = \sum_{l=1}^{\nes} \tr[\rho E_{l}] E_{l},
\end{equation}
which is invertible if $E$ spans the whole space of observables~\cite{shadow_huang_2020,povm_shadow_nguyen_2022}.
The snapshot can be thought of as single-shot-estimator of the prepared state $\rho$.
In an experiment one repeats the above single-shot procedure $\nshot$ times, which produces a collection of
outcomes $\{l_1,l_2,\ldots,l_{\nshot}\}$. 
Accordingly, a collection
of snapshots can be constructed 
\begin{align*}
	S(\rho, \nshot) = \{\hat{\rho}_{l_1}, \hat{\rho}_{l_2}, \dots \hat{\rho}_{l_{\nshot}} \},
\end{align*}
which is called a \emph{classical shadow} of $\rho$.
The classical shadow allows us to obtain an unbiased estimate for the density operator in the sense $\rho  = \mathbb{E}_{l}[\hat{\rho}_l]$.

Crucial to the advantage of shadow tomography is that when the measurement $E$ consists of independent measurements on individual qubits, the snapshots $\hat{\rho}_l$ also factorise into a tensor product over the qubits. It is therefore sufficient to store single-qubit tensoring factors of $\hat{\rho}_l$, instead of the exponentially large matrix itself~\cite{shadow_huang_2020}. Functions of the density operator with appropriate locality, such as correlation functions or the Rényi entropy, can also be efficiently estimated~\cite{shadow_huang_2020}.
As an example, for experimentally-friendly case of randomised noiseless Pauli basis measurements on the qubits, the snapshot corresponding to $l=(j,b)$ is given explicitly by 
\begin{align} \label{eq:snapshot}
	\hat{\rho}_{l} =   \bigotimes_{i=1}^N \Big[ 3 (Q^{(i)}_{j})^{\dagger} \vert b^{(i)}\rangle \langle b^{(i)} \vert  Q^{(i)}_{j} - \openone \Big].
\end{align}
Above $b^{(i)}$ is the $i^{th}$ bit of the $N$-qubit measurement outcome bitsring $b$, 
and $Q^{(i)}_{j}$ is the $i^{th}$ single-qubit basis transformation
in the applied $N$-qubit Pauli basis transformation $Q_j$.
In the following, we focus on this practically-pivotal randomised Pauli-measurement scheme. 
However, our general formalism can immediately be applied to
other unitary ensembles such as matchgates~\cite{wan_matchgate_2022},
Clifford circuits~\cite{shadow_huang_2020} and beyond~\cite{bertoni2022shallow}.

\subsection{Mitigating readout errors\label{sec:readout}}

An advantage of having introduced classical shadows through generalised measurements (POVMs)
is that it is now straightforward to incorporate readout-error mitigation techniques~\cite{povm_shadow_nguyen_2022}. 
In particular, readout errors refer to the classical process of incorrectly assigning the labels $b$ to the measurement outcome.
While in ion-trap devices readout errors may  not be significant, i.e., below error levels of gate
operations~\cite{PRXQuantum.2.020343,moses2023race}, in solid-state devices these errors can be quite substantial
and can be on the order of several percent~\cite{kim2023evidence}.
We illustrate the approach by considering a simple readout-error model where entries in the bitstring $b$ undergo random
and independent bit flips ($i^{th}$ bit is flipped as $0 \rightarrow 1$ with probability $\alpha_i^{+}$ and as $1 \rightarrow 0$ with probability $\alpha_i^{-}$)
-- while existing techniques for addressing correlated readout errors as well as gate
errors in the shadow basis rotations
are indeed similarly applicable~\cite{maciejewski2020mitigation,detector_tomoIBM_2019,cai2022quantum,povm_shadow_nguyen_2022}.
As a consequence, the idealised measurement operators $\ketbra{0}{0}$ and $\ketbra{1}{1}$ on the $i^{th}$ qubit
are replaced according to the readout error model, in the present case by $(1-\alpha_i^{+}) \ketbra{0}{0}+ \alpha_i^{-} \ketbra{1}{1}$
and by $\alpha_i^{+} \ketbra{0}{0}+(1-\alpha_i^{-}) \ketbra{1}{1}$,
which then allows us to explicitly build the effects in \cref{eq:idealised_measurements} and invert the measurement channels in \cref{eq:measurement_channel}.

For example, we can analytically obtain a simple formula for the Pauli measurement snapshot from~\cref{eq:snapshot} 
for the specific readout-error model when $\alpha_i^{+}=\alpha_i^{-}=:\alpha_i$ as
\begin{equation*}
\hat{\rho}_{l} = \bigotimes_{i=1}^N \Big[ \frac{3}{1-2\alpha_i} (Q^{(i)}_{j})^{\dagger} \vert b^{(i)}\rangle \langle b^{(i)} \vert  Q^{(i)}_{j} - \frac{1+\alpha_i}{(1-  2\alpha_i)} \openone \Big].
\end{equation*}
In the more realistic case when $\alpha_i^{+} \ne \alpha_i^{-}$ the snapshots can still be computed
straightforwardly by numerically (rather than analytically) inverting the single-qubit channel in~\cref{eq:measurement_channel},
model-free readout-error mitigation is also applicable~\cite{model_free_readout_error_mit}.
In conclusion, as measurement-error mitigation techniques are completely decoupled from mitigating state-preparation errors,
our mathematical theorems will quantify the sample complexity of the latter,
while we demonstrate in numerical simulations that it is indeed straightforward
to combine measurement-error mitigation techniques with the present approach.


\section{Probabilistic error cancelled shadows\label{sec:pec}}

While PEC has been used in the literature to remove the bias in expected-value measurements~\cite{cai2022quantum},
here we apply it to classical shadows to obtain an efficient, classical representation of
the entire ideal, noise free state $\rho_{id}$.
While this procedure even allows us to estimate the full density matrix $\rho_{id}$,
we will focus on efficient practical applications such as \emph{simultaneously predicting many properties} of $\rho_{id}$.
At a technical level PEC shadows is a combination of two random processes, i.e.,
sampling circuit variants $\mathcal{G}_{\underline{k}}$ and
sampling the bitstrings and the basis transformations that form a shadow.

\subsection{Probabilistic error cancellation}
PEC is one of the most broadly studied error mitigation techniques~\cite{gambetta_error_mitig,practical_QEM,cai2022quantum}
and indeed has been implemented experimentally~\cite{cai2022quantum,kandala2019error,kim2023scalable, o2022purification};
It is performed by decomposing the channel $\mathcal{U}$ of an ideal unitary gate 
into a linear combination of noisy physical gate operations $\mathcal{G}_k$ as
$	\mathcal{U}  = \sum_{k} \gamma_k \mathcal{G}_k$.
Negative quasiprobabilities $\gamma_k < 0$  are required to formally implement
the inverse of a noise channel. Thus the above operation is nonphysical, similarly
as the inverse measurement channels of the shadow protocols are nonphysical operations.
For this reason, PEC only applies the decomposition in classical post-processing,
at the level of expected values
and allows us to compute ideal expected values as a linear combination
of noisy ones as
$\sum_{k} \gamma_k \tr [ O \mathcal{G}_ k |0\rangle\langle0| ]$.

Let us first give an overview of efficient methods that
have been developed to accurately identify the coefficients $\{\gamma_k\}$ for the quasi-probability decomposition~\cite{strikis2021learning, PhysRevResearch.3.033098, berg2022probabilistic, montanaro2021error, cai2022quantum}.
The simplest such approach exploits that noise models are approximately local and one can 
thus efficiently characterise the local noise channel of each gate and invert them classically.
Recent advances allow for non-local noise models of the form $\Lambda^G = e^{\mathcal{L}}$
to be efficiently learned for the case of sparse Pauli operations $\mathcal{L}$
resulting in a trivial inverse of the channel as $\Lambda^{G-1}$~\cite{berg2022probabilistic}.
Making the assumption that $\mathcal{U}$ is supported only on the space spanned by the noisy operations $\mathcal{G}_k$,
one then randomly applies circuit variants that implement the operations $\mathcal{G}_k$ in the inverse noise channel.
Under this assumption, theoretical guarantees have been derived on the sample complexity of learning
the error model~\cite{berg2022probabilistic}. Given the noise model can be learned prior to the experiment
with a constant overhead, we assume in this work that a sufficiently high-precision estimate
is available~\footnote{For example, one can consider a per-gate error rate $p_{err} = 10^{-2}$ and assume the gate's
error model is characterised to a precision $\epsilon = 10^{-4}$.
To avoid exponential blow-up of the PEC sampling overhead, the number of applications of this gate is limited to
$N_{gate} \leq p_{err}^{-1}$
and over the course of $N_{gate}$ applications the error propagation due to the imperfect error model 
can be approximated as $(1 -  \epsilon)^{N_{gate}} =  1 - \epsilon/ p_{err} + O (\epsilon ^2)$
confirming dependence on the ratio of the two error rates
--
which are typically chosen to be several (presently 2) orders of magnitude different.
}

Indeed a considerable experimental challenge is posed by the possible drift of the error models over time
which may necessitate the repetition of the learning procedure every so often further increasing its
sample budget.
Let us finally note that our scope goes beyond mitigation
of physical gate errors and the formalism developed here can be immediately
 applied to other scenarios, such as the following. a) Overcoming finite rotation-angle resolution
whereby the quasiprobability decomposition is known exactly~\cite{koczor2023probabilistic}
b) Mitigating logical errors in early-fault tolerant devices
whereby the dominant source of noise may arise from imperfect
magic-state distillation but noise characterisation and mitigation 
are effectively identical to the presently considered
formalism~\cite{PhysRevLett.127.200506,PhysRevLett.127.200505,PRXQuantum.3.010345}.

We consider that an ideal quantum state $\rho_{id} := \mathcal{U}_{circ} |0 \rangle \langle 0 |$ is prepared by an ideal circuit
of $\nu$ gates whose channel we denote as
$\mathcal{U}_{circ} = \mathcal{U}_\nu \cdots \circ\mathcal{U}_2 \circ \mathcal{U}_1$.
By introducing the vector notation $\underline{k} = (k_1, k_2, \dots k_\nu)$,
we can compactly represent the decomposition 
of this circuit  into noisy gate sequences as
$\mathcal{U}_{circ}  = \sum_{\underline{k}} g_{\underline{k}} \mathcal{G}_{\underline{k}}$.
Here the index $\underline{k}$ indexes all possible gate sequences
as
\begin{equation} \label{eq:product}
	\mathcal{G}_{\underline{k}} = \mathcal{G}^{(1)}_{k_1} \mathcal{G}^{(2)}_{k_2}  \cdots \mathcal{G}^{(\nu)}_{k_\nu},
	\quad 
	\quad
	g_{\underline{k}} = \gamma^{(1)}_{k_1} \gamma^{(2)}_{k_2}  \cdots \gamma^{(\nu)}_{k_\nu},
\end{equation}
and as shown above the corresponding quasiprobabilities $g_{\underline{k}}$ factorise
(the superscript indexes individual gates, e.g., $\mathcal{G}^{(1)}_{k_1}$ stands for the decompositon of $\mathcal{U}_1$).
We now define the quasiprobability decomposition of a quantum circuit.
\begin{definition}\label{def:decompose}
	We define the quasiprobability decomposition of an
	ideal circuit $\mathcal{U}_{circ}$  via the set $G := \lbrace (g_{\underline{k}},\mathcal{G}_{\underline{k}}) \rbrace_{\underline{k}}$.
	We also define the associated probability distribution $p(\underline{k}) := |g_{\underline{k}}| / \lVert g\lVert_1$
	and here the norm factorizes as $\lVert g\lVert_1 =  \prod_{k=1}^{\nu} \lVert \gamma^{(k)} \rVert_1$
	into a product of individual norms.
\end{definition}
The above quasiprobability decomposition has been used for estimating the ideal expected value of an observable $O$,
 $\tr[ O \mathcal{U}_{circ} |0 \rangle \langle 0 | ]$,
by randomly sampling the noisy expected values $\sign(g_{\underline{k}}) \tr[ O \mathcal{G}_{\underline{k}} |0 \rangle \langle 0 | ]$
according to the probability distribution $p(\underline{k})$
and linearly combining them in a classical post-processing step~\cite{gambetta_error_mitig,practical_QEM,cai2022quantum}.
The norm $\lVert \gamma^{(k)} \rVert_1$ can be evaluated straightforwardly for any probabilistic
	error model:
	Assuming that during the execution of the $k^{th}$ gate an error happens with probability
	$p_k$, e.g., Pauli errors, we obtain the single-gate norm $ \lVert \gamma^{(k)} \rVert_1 = (1+p_k)/(1-p_k)$~\cite{cai2022quantum}.
	Thus, the cost of error mitigation---as the product of these individual norms from \cref{def:decompose}---grows  as
	$\lVert g\lVert_1 \in O(e^{2\xi})$ with the expected number $\xi = \sum_k p_k$ of errors in the full circuit,
	rendering the approach impractical when $\xi \gg 1$~\cite{gambetta_error_mitig,practical_QEM,cai2022quantum}.

\subsection{Details of the protocol}
We start by applying the PEC protocol in a more general setting such that
the quasiprobability decomposition allows us to obtain an unbiased estimator
of the full density matrix.
\begin{lemma}\label{lemma_PEC_unbiased}
	Given a quasiprobability decomposition $G$ from \cref{def:decompose},
	by sampling the noisy circuits $\mathcal{G}_{\underline{k}}$ according to the probability distribution
	$p(\underline{k})$ we obtain an unbiased estimator
        of  the ideal density matrix $\rho_{id} := \mathcal{U}_{circ} |0 \rangle \langle 0 |$ as
	\begin{equation}\label{eq_estimator_rhoid}
		\hat{\rho}_{id} =  \lVert g\lVert_1 \, \mathrm{sign} (g_{\underline{k}}) \mathcal{G}_{\underline{k}} \vert 0 \rangle \langle 0 \vert
	\end{equation} 
	in the sense that $\mathbb{E}_{\underline{k}} [ \hat{\rho}_{id} ]  = \rho_{id}$.
\end{lemma}
The above estimator has a clear operational meaning:
(1) choose a multi-index $\underline{k}$ randomly according to the probability  distribution $p(\underline{k})$
and run the noisy quantum circuit $\mathcal{G}_{\underline{k}}$;
(2) the output state $\mathcal{G}_{\underline{k}} |0 \rangle \langle 0 |$ is a density matrix
that we multiply by $\mathrm{sign}(g_{\underline{k}})$ and with the norm $\lVert g\lVert_1$;
(3) formally, the mean of these matrices is an estimate of the ideal density matrix $\rho_{id}$.

Regrettably, the above protocol is purely formal as the multiplication with negative quasiprobabilities
is non physical and could only be achieved in post-processing, e.g., after fully reconstructing the density matrix.  
We thus exploit classical shadows as a powerful tool for obtaining an efficient classical description of the
states which can then be naturally assigned negative quasiprobabilities in classical post-processing.
Indeed,  snapshots are not physical density matrices either, as is apparent in \cref{eq:snapshot}.
We now state our protocol that serves as an unbiased estimator of the ideal state.
\begin{theorem}[PEC shadows]\label{theo:mitigated_shadow}
Given a quasiprobability decomposition $G$ of the ideal circuit  $\mathcal{U}_{\text{circ}}$
from \cref{def:decompose},
and a classical shadow protocol with the POVM measurement $E$ from \cref{eq:idealised_measurements},
we define PEC shadows as the set
$\emshadow := \lbrace (g_{\underline{k}},\mathcal{G}_{\underline{k}}, E_l) \rbrace_{\underline{k}, l}$
and define the corresponding PEC snapshot as
\begin{align}\label{eq:mitigated_shadow_estimator}
    \hat{\rho}_{id} := \hat{\rho}_{\underline{k}, l}   =  \lVert g \rVert_1 \, \mathrm{sign}(g_{\underline{k}}) C_{E}^{-1} (E_{l}).
\end{align}
We will often use the notation $\hat{\rho}_{id}$ to abbreviate $\hat{\rho}_{\underline{k}, l} $ as it is an unbiased estimator of the ideal density matrix $\rho_{id}$
such that $\mathbb{E}[ \hat{\rho}_{id} ] = \underset{\underline{k},l}{\mathbb{E}} [\hat{\rho}_{\underline{k}, l} ]  = \rho_{id}$.
\end{theorem}
Above the averaging $\mathbb{E}[\cdot]$ happens not only 
over the effects $E_l$ indexed by $l$ (all basis transformations and all measurement outcomes),
but additionally we average over all circuit variants indexed by $\underline{k}$.
The reason is that the measurement $E=\lbrace E_{l} \rbrace_{l}$ is not performed on a fixed input density matrix $\rho$
as in conventional shadows
but rather on the quasiprobability decomposition of the ideal state
$\rho_{id} \propto \mathcal{G}_{\underline{k}} \vert 0 \rangle \langle 0 \vert$.
Let us now summarise the resulting experimental protocol.
\begin{itemize}[leftmargin=*]
	\item choose randomly a multi-index $\underline{k}$ according to the probabilities
	$p(\underline{k})$ and store the $\sign ( g_{\underline{k}} )$
	
	\item choose uniformly randomly a unitary rotation $Q_j \in \mathcal{Q}$
	and store its index $j$
	
	\item  execute in a quantum computer the gate sequence $\mathcal{G}_{\underline{k}}$,
		the unitary rotation $Q_j$, perform a measurement in the standard basis
		and finally register its outcome $b$
    
	\item each stored index $(\sign [ g_{\underline{k}} ], j, b)$
	uniquely identifies a classical snapshot   
	$\hat{\rho}_{\underline{k}, l}   =  \lVert g \rVert_1 \, \mathrm{sign}(g_{\underline{k}}) C_{E}^{-1} (E_{l})$
	where recall that $E_l$ is a POVM effect with the index $l = (j,b)$ from \cref{eq:idealised_measurements}  
	
	\item repeat the procedure and collect $\nshot$ classical snapshots to build a classical shadow of the ideal state
	$S(\rho_{id}, \nshot) = \{   (\hat{\rho}_{id})_{1}, (\hat{\rho}_{id})_2, \dots, (\hat{\rho}_{id})_{\nshot} \}$
\end{itemize}

The classical dataset $S(\rho_{id}, \nshot)$ can then be classically post-processed offline
and we detail explicit algorithms for
predicting local properties in \cref{sec:post_proc}. 

Note that PEC shadows produce a distribution of snapshots that is different than directly applying
conventional shadows to a noise-free state $\rhoid$, albeit with an identical mean.
The reason is that each circuit variant $\mathcal{G}_{\underline{k}}$ in \cref{eq_estimator_rhoid}
yields a different distribution of classical snapshots.
For example, in the next section we prove bounds 
on variances of PEC shadows and find that they are indeed increased compared to conventional shadows applied
directly to $\rhoid$.


\subsection{Rigorous performance guarantees}
We first consider the pivotal practical application as predicting error mitigated expected values
of observables $O$ via the estimator $\hat{o} = \tr[O \hat{\rho}_{id}]$.
A key observation is that in error mitigation techniques the ability to predict
noise-free expected values comes at the cost of an
increased statistical variance which implies an increased number of circuit repetitions.
We now bound the variance of any operator's expected value.

\begin{lemma}[variance of linear properties]\label{lemma:expval_variance}
Given an observable $O$ and the PEC snapshot $\hat{\rho}_{id}$ 
from \cref{theo:mitigated_shadow}, the variance of
$\hat{o} = \tr[O \hat{\rho}_{id}]$ can be upper bounded as
\begin{align}
    \var[\hat{o}]  \leq \lVert g \rVert_1^{2} \, \lVert O \rVert_{E}^{2},
\end{align}
where $\lVert \cdot \rVert_{E}^{2}$ is the shadow norm of the observable $O$
as defined in \cref{lemma:shadow_norm}.
When $O$ is a $q$-local Pauli string and we use Pauli basis measurements then
$\lVert O \rVert_{E}^{2} = 3^q$ as we detail in \cref{lemma:shadow_norm}.
\end{lemma}
We explain in \cref{app:numerics} that we can account for the cost of readout-error mitigation
	via the above shadow norm, which becomes
$3^q(1{-}2 \alpha)^{-2q}$ when considering a simple readout-error model with probability at most $\alpha$.
Observe that the above variance depends on two factors:
The first one is the squared shadow norm $\lVert O \rVert_{E}^{2}$ which
determines the sample complexity of 
conventional shadows~\cite{shadow_huang_2020};
The second factor is a multiplicative term $\lVert g \rVert_1^{2}$
which accounts for the well-known measurement overhead associated with
the conventional PEC protocol~\cite{gambetta_error_mitig,practical_QEM,cai2022quantum}

We defer further discussion to \cref{app:shadow_norm} where we
also explain how the shadow norm depends on the unitary ensemble $\mathcal{Q}$:
While we only state explicitly the shadow norm for the practically most important ensemble of Pauli basis measurements,
we note that bounds for other ensembles are immediately available in the literature~\cite{wan_matchgate_2022,shadow_huang_2020,bertoni2022shallow}.
Furthermore, we also explain in \cref{app:variance} that the above bound is
expected to be pessimistic due to an even more significantly overestimated constant prefactor
than in conventional shadows.

Following the approach of~\cite{shadow_huang_2020} 
we use concentration properties of the median of means estimator to derive rigorous sample
complexities:
for the simultaneous prediction of many observables $O_{1}, \dots, O_{M}$ 
we exponentially suppress statistical outliers
by splitting the PEC shadows $S(\rho_{id}, \nshot)$ into independent batches
and then computing a median of the means
as we detail in \cref{app:proof_m_properties}.
The resulting bounds depend on two performance metrics as
the accuracy $\epsilon$ and the success probability $\delta$.
\begin{theorem}[informal summary]\label{theo:m_properties}
Given the PEC shadows 
$\emshadow := \lbrace (g_{\underline{k}},\mathcal{G}_{\underline{k}}, E_l) \rbrace_{\underline{k}, l}$
from \cref{theo:mitigated_shadow} we want to simultaneously estimate expected values of $M$ operators in the
ideal state as $\tr[O_{1} \rho_{id}] \dots \tr[O_{M} \rho_{id}]$.
Using a median of means estimator, the number of shots required to achieve performance
parameters $\epsilon, \delta \in [0,1]$ is 
\begin{align}{\label{eq:error_bound}}
    \nshot  = 32 \epsilon^{-2}
     \, \,
      \log (\frac{M}{\delta})
      \, \,
      \lVert g \rVert_1^{2}
      \, \, 
      \max_{1 \leq k \leq M}{\lVert O_{k} \rVert_{E}^{2}},
\end{align}
where we use the largest shadow norm $\lVert O_{k} \rVert_{E}^{2}$.
Refer to \cref{theo:m_prop_formal} for a formal statement of this theorem.
\end{theorem}

\begin{figure}
	\centering
	\includegraphics[width=0.483\textwidth]{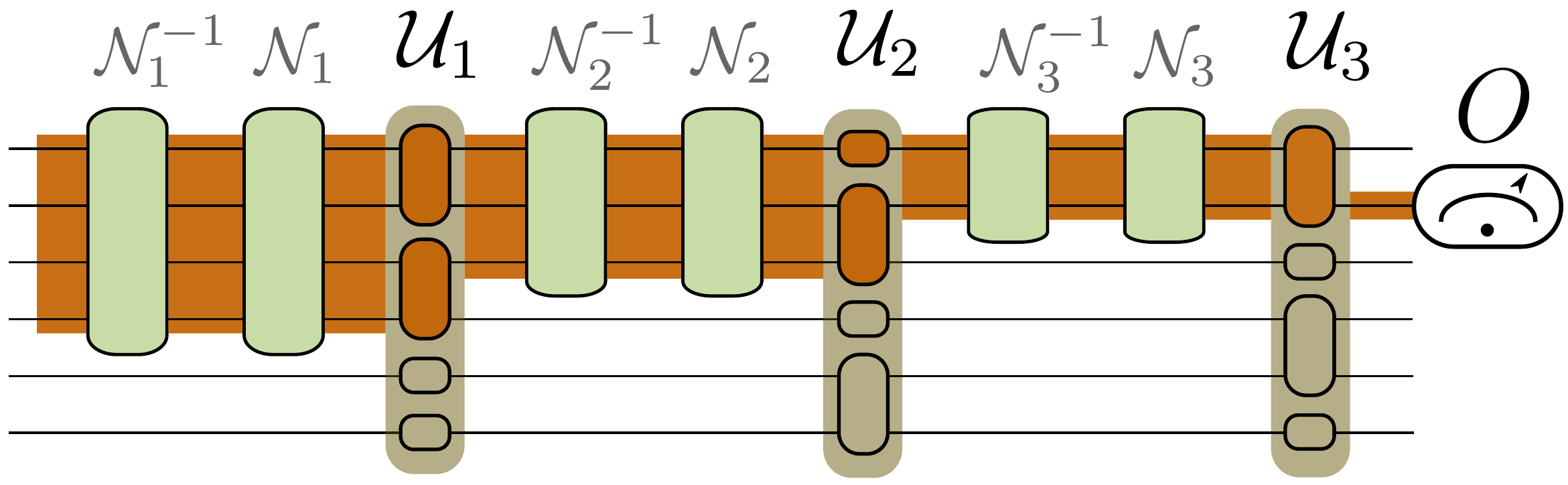}
	\caption{Illustration of the light cone of an observable $O$
        which is represented by the measurement apparatus
        and only acts nontrivially on the second (from top) qubit.
        The orange area indicates the qubits which are contained in the light cone $\mathcal{I}$ of the observable
        with respect to the ideal quantum circuit (orange boxes)
        $\mathcal{U}_{3} \mathcal{U}_{2} \mathcal{U}_{1}$.
        To simplify derivations we assume the
        gate noise channels $\mathcal{N}_{k}$ are local (light green boxes) such that 
        they are contained within the lightcone $\mathcal{I}$
        but our results can be extended to non-local models via~\cite{locality_error_mitigation_ibm_2023}.
        }
	\label{fig:light_cone}
\end{figure}


Finally, we consider predicting non-linear properties of the state of the form $\tr[ O (\rhoid)^m ]$.
Following ref.~\cite{shadow_huang_2020} we use the fact that a polynomial function in the quantum state
can be written as a linear function in tensor products of the state, for example,
$\tr[O(\rhoid)^{2}] = \tr[\Tilde{O} \rhoid \otimes \rhoid ]$ with $\Tilde{O} = \mathrm{SWAP} (O \otimes \openone)$
where the SWAP operator swaps the two copies.
In \cref{app:nonlin} we detail our construction using U-statistics to derive unbiased
estimators in terms of the classical snapshots, e.g., for $m=2$ we
select all distinct pairs of snapshots $(\hat{\rho}_{id})_{i}  \otimes  (\hat{\rho}_{id})_{j}$ with
$i \neq j$. We can bound the variance of any non-linear property as follows.
\begin{theorem}[variance of nonlinear properties] \label{theo:nonlin}
Given our PEC snapshots $\hat{\rho}_{id}$ from \cref{theo:mitigated_shadow}
we can estimate polynomial properties of degree $m$ of the ideal state $\rhoid$
via U-statistics of tensor products of all distinct snapshots.
The number of samples required to predict the non-linear property
scales as $\nshot \in \mathcal{O}(\vert \vert g \vert \vert_{1}^{2m} / \epsilon^2)$
for a desired accuracy $\epsilon$.
Refer to \cref{app:nonlin} for details.
\end{theorem}
One can similarly apply a median of means estimator to enable simultaneous prediction of
many non-linear properties: We detail an explicit protocol in the next section for
simultaneously estimating many local Rényi entropies. Note that measurement overhead $\vert \vert g \vert \vert_{1}^{2m}$
grows with the $2m^{th}$ power of the quasiprobability norm consistent with our effective construction
of $m$ copies of the original noisy circuit which leads to an effective $m$-fold increase
in the number of noisy gates.


\subsection{Classical post-processing algorithms\label{sec:post_proc}}

In this section we summarise reconstruction algorithms
for the practically pivotal scenario of Pauli basis measurements from \cref{sec:widealised_measurements}.
\begin{algorithm}[local Pauli strings]\label{alg:reconstruct_algo}
	The expected value of a $q$-local Pauli observable $P$ in a PEC snapshot
	can be calculated analytically  as
	\begin{equation*}
		\tr[ P \hat{\rho}_{\underline{k}, l} ] = \lVert g \rVert_1 \, 3^q \, \sign(g_{\underline{k}}) f(b, Q_j).
	\end{equation*}
	The reconstruction algorithm iterates over all snapshots in the shadow $S(\rho_{id}, \nshot)$
	and calculates the median of means of the above expression
	using a number of batches provided by the user.
	Here $f(b,Q_j) \in \{\pm1, 0\}$ results in $0$ if the measurement bases
	in $Q_j$ are incompatible with $P$ and $\pm 1$ if the measurement bases are compatible with $P$
	while the sign is determined by the bitstring $b$.
        The algorithm has runtime $\mathcal{O}( q \nshot )$.
	\qed
\end{algorithm}
We defer the detailed derivation to \cref{app:reconstruct_algo}.
Note that the above error mitigated reconstruction algorithm deviates from that of~\cite{shadow_huang_2020} as
the individual snapshot outcomes are multipiled with 
the norm $\lVert g \rVert_1$ and signs of the quasiprobabilities $\mathrm{sign}(g_{\underline{k}})$.

Since we reconstruct $q$-local Pauli observables, we can significantly reduce the sample variacne via light-cone
arguments~\cite{hierarchy_light_cones_tran_2020,locality_error_mitigation_ibm_2023}.
In \cref{fig:light_cone} we illustrate the light cone that an observable creates
with respect to the ideal unitary circuit
$\mathcal{U}_{circ}$.
To simplify the following arguments we assume local noise models to guarantee the same light cone is
valid for all gate sequences $\mathcal{G}_{\underline{k}}$, however,
it is straightforward to extend the arguments to non-local noise following~\cite{locality_error_mitigation_ibm_2023}.

We observe that for each gate that is not within the light cone of the observable $P$
we can ``turn off'' PEC thereby not wasting our measurement budget on mitigating noisy
gates that do not affect our observable.
\begin{algorithm}[light cones]\label{alg:light_cone}
	Given a $q$-local Pauli string $P$
	we define the set of indexes of all gates in the light cone of the obserable as
	$\mathcal{I} := \{ l \, | \, \mathcal{U}_l \, \text{is in the light cone of $P$} \}$.
     We then simply use \cref{alg:reconstruct_algo} with a modified set of quasiprobabilities
     from \cref{def:decompose} as
	\begin{equation*}
		\lVert \tilde{g} \rVert_1 = \prod_{l \in \mathcal{I} } \lVert \gamma^{(l)} \rVert_1,
		\quad \text{and}  \quad
		\sign( \tilde{g}_{\underline{k}} ) = \prod_{l \in \mathcal{I} } \sign \gamma^{(l)}_{k_l}.
	\end{equation*}
   The algorithm has the same asymptotic runtime $\mathcal{O}( q \nshot )$ as \cref{alg:reconstruct_algo}
   and only incurs a negligible
   preprocessing time to determine the index set $\mathcal{I}$ specifically for each $P$.
	\qed
\end{algorithm}
Refer to \cref{app:light_cone} for a derivation.
The measurement cost $\lVert \tilde{g} \rVert_1$ is thus determined by the
number of gates in the light cone of $P$
rather than by the total number of gates $\nu$.
Imagine, for example, noisy quantum gates with $\lVert \gamma^{(l)} \rVert_1 = 1{+}p$;
the measurement cost is determined by $(1{+}p)^{|\mathcal{I}|}$ as opposed to 
the worst case 
$\lVert \tilde{g} \rVert_1  = (1{+}p)^\nu$ where $\nu$ is the total number of noisy
gates as detailed in Ref.~\cite{locality_error_mitigation_ibm_2023}.
A significant advantage of this procedure is that it does not require one to modify the 
experimental protocol, i.e., the noise in all gates can be mitigated in the shadows.

Finally, we consider estimating Rényi entropies via the purities $\tr\left(\rho_Q^2\right)$ as $R_Q := -\log \tr\left(\rho_Q^2\right)$ where $\rho_Q$
is the reduced density matrix of the subsystem $Q$.
\begin{algorithm}[local purities]\label{alg:entropy}
    Given a subsystem as the set of qubits $Q =\left\{q_1, ..., q_m\right\}$,
    an unbiased estimator for the respective purity is obtained as
    \begin{equation}
   		\tr\left(\hat{\rho}_Q^2\right)  = \lVert g \rVert_1^2 \sign(g_{i})\sign(g_{j}) f( i, j, Q).
    \end{equation}
	Here $i$ and $j$ abbreviate indexes of the snapshots
    as, e.g., $i = (\underline{k}, l)$ and $i \neq j$.
    The algorithm iterates over all distinct pairs of snapshots in the shadow $S(\rho_{id}, \nshot)$
    and calculates  the median of means of the above expression.
    The factors $f( i, j, Q)$ depend only on whether the measurement bases and outcome
    bitsrings are identical within the subsystem $Q$.
    The algorithm has a runtime $\mathcal{O}( |Q| \nshot^2 )$.
   	\qed
\end{algorithm}
We defer the detailed derivation to \cref{app:reconstruct_algo}.
Note that the runtime is linear in the subsystem size and quadratic in the number of shots.
For sufficiently small subsystems and large numbers of shots it might be preferred to
use the exponentially $\mathcal{O}( 4^{|Q|} \nshot )$ scaling algorithm of \cite{shadow_huang_2020}.

\section{Further error mitigation techniques\label{sec:other_mitigation}}

\subsection{Error extrapolated shadows}
The key idea behind zero-noise extrapolation resides in the possibility of
increasing the noise in the circuit and extrapolating expected values back to the case of zero noise.
The approach is intuitive to use, requires less resources than PEC but yields a biased estimator.
A non-trivial aspect, however, is choosing the correct model function for the extrapolation
which has been extensively discussed in the literature~\cite{cai2022quantum,practical_QEM, kandala2019error};
typical models include a linear function, an exponential function or a linear combination of multiple exponentials.

We consider extrapolation as a means for mitigating errors in properties
extracted from classical shadows. The key ingredient we require is the ability to generate a
collection of shadows at different noise strengths $S(\rho_{p_0}, N_s), ..., S(\rho_{p_n}, N_s)$
such that $p_k \geq p_0$ and $p_0$ is the device's lowest possible noise strength.
These shadows enable us to extract the expected values $f_{m}(p) = \tr [ O_{m} \rho_p ]$ 
at a given noise level $p$.
By fitting a suitable model function $\tilde{f}_m(p)$, e.g., a linear model, to this
dataset we can approximate ideal properties of the state using an extrapolation via the limit
\begin{equation*}
 \tr [ O_m \rho_{id} ]  \approx	 \lim_{p \rightarrow 0}  \tilde{f}_m(p).
\end{equation*}

While we could certainly leverage existing techniques for
physically increasing noise rates in a circuit to obtain $S(\rho_{p}, \nshot)$~\cite{kandala2019error,cai2022quantum},
we can also exploit the power and flexibility
of the previously derived PEC shadow approach:
Instead of considering the quasiprobability representation of the ideal circuit in \cref{def:decompose}
we can rather decompose the noise-boosted circuits as
$\mathcal{G}_{circ}(p)  = \sum_{\underline{k}} g_{\underline{k}}(p) \mathcal{G}_{\underline{k}}$
with non-negative probabilities $g_{\underline{k}}(p)$.
For example,  in the case of local depolarising noise the circuit variants $\mathcal{G}_{\underline{k}}$
are simply obtained by randomly inserting Pauli $X$, $Y$ or $Z$ operations after each noisy gate
with probabilities $p{-}p_0$.
Furthermore, Lindblad-Pauli learning directly gives access to the continuous set of circuits
$\mathcal{G}_{circ}(p)$~\cite{berg2022probabilistic}.

Let us now state a corollary to \cref{theo:mitigated_shadow} that allows us to obtain the
shadows of error boosted states $\rho_p := \mathcal{G}_{circ}(p) \rho_{ref}$.
\begin{corollary}[error-boosted shadows]\label{cor:noise-boost}
	We consider the parametric quasiprobability decomposition $G$ as noise-boosted circuits  $\mathcal{G}_{circ}(p)$
	with $p \geq p_0$.
	The PEC shadows $\emshadow:= \lbrace (g_{\underline{k}}(p),\mathcal{G}_{\underline{k}}, E_l) \rbrace_{\underline{k}, l}$
	from \cref{theo:mitigated_shadow} result in the simplified 
	snapshots as $ \hat{\rho}_{p} := \hat{\rho}_{p, (\underline{k},l)} = C_{E}^{-1} (E_{l})$
	due to $\sign (g_{\underline{k}}(p) ) = +1$ and $\lVert g(p) \rVert_1 = 1$.
	It follows that $\hat{\rho}_{p}$ is an unbiased estimator of the noise-boosted
	density matrix $\rho_{p}$
	such that $\mathbb{E}[\hat{\rho}_{p}] =\underset{\underline{k},l}{\mathbb{E}} [\hat{\rho}_{p, (\underline{k},l)}]  = \rho_{p}$.
\end{corollary}
A significant advantage in boosting noise via $p \geq p_0$
rather than reducing it is that now every
quasiprobability is non-negative $g_{\underline{k}}(p) \geq 0$ and thus we do not
incur a measurement overhead in \cref{theo:m_properties} via $\lVert g(p) \rVert_{1} = 1$.
Nevertheless, the extrapolated value indeed suffers from an increased variance
which implies an increased number of samples
and details can be found in the literature~\cite{cai2022quantum}. 

Note that the above scheme can be applied beyond the estimation of expectation values.
For instance, one can in principle use shadows to reconstruct partial density
matrices $\hat{\rho}_p$ at different noise strengths $p$ and apply ZNE to individual matrix entries. 
However, note that ZNE might require different kinds of model
functions $f(p)$ for different properties, e.g., non-linear models for
predicting non-linear properties of the state.
In contrast, the great advantage of PEC shadows is that it provides
an unbiased estimator for the entire quantum state.

\subsection{Symmetry verified shadows}
Symmetry verification is another leading quantum error mitigation
technique~\cite{mcardleErrorMitigatedDigitalQuantum2019,bonet-monroigLowcostErrorMitigation2018};
It exploits that often the ideal states to be prepared $\rhoid$ are pure states that
obey certain problem specific symmetry group operations described by $S \in \mathbb{S}$.
The fact that the ideal state is symmetric 
then implies that it ``lives in'' the subspace defined by the projection operator
\begin{align*}
	\Pi_{\mathbb{S}} = \frac{1}{|\mathbb{S}|} \sum_{S \in \mathbb{S}} S,
\end{align*}
which satisfies $\Pi_{\mathbb{S}}^2 = \Pi_{\mathbb{S}}$. 

Given a noisy state $\rho$, one might be able to measure the above symmetries (in fact their generators are sufficient)
via, e.g., Hadamard-test circuits, and retain only circuit
runs that produce the correct symmetry outcomes~\cite{cai2022quantum,PhysRevA.105.022441}. Such post-selection
projects the noisy state back into this symmetry subspace
producing the effective output state as
$\rho_{\textrm{sym}} = \Pi_{\mathbb{S}}\rho \Pi_{\mathbb{S}} / \Tr(\Pi_{\mathbb{S}} \rho)$.
We can apply conventional shadow tomography to this symmetry-verified state $\rho_{\textrm{sym}}$
thereby effectively obtaining error mitigated shadows, i.e., an unbiased estimator of $\rho_{\textrm{sym}}$.
The sampling overhead of this post-selection technique is $\Tr(\Pi_{\mathbb{S}} \rho)^{-1}$
the inverse of the fraction of circuit runs that pass the symmetry verification process. 

Instead of post-selection, we can also perform symmetry verification at the
post-processing stage. Suppose we are interested in the expectation value of
the ideal state with respect to the target observable $O$,
the target expectation value can be written as
\begin{align*}
	\Tr(O\rho_{\textrm{sym}}) = \frac{\Tr(O\Pi_{\mathbb{S}}\rho \Pi_{\mathbb{S}})}{\Tr(O\Pi_{\mathbb{S}} \rho)} = \frac{1}{|\mathbb{S}|} \frac{\sum_{S,S' \in \mathbb{S}}\Tr(SOS'\rho)}{\sum_{S \in \mathbb{S}}\Tr(S \rho)}
\end{align*}

Conventional shadow tomography can be used well in practice
for estimating  $SOS'$ and $S$ for all $S, S' \in \mathbb{S}$
when the symmetries are sufficiently local, i.e., they are supported on at most
weight-$s$ Pauli operators.
Then, given a Pauli observable $O$ of weight at most $q$, 
the effective observable $SOS'$ is then at most of weight-$(2s{+}q)$.
However, the sample complexity of conventional shadows with Pauli measurements
grows exponentially with the weight of the Pauli string and it is thus  crucial
that the total weight $2s{+}q$ be reasonably small.

\begin{figure*}[tb]
	\centering
	\includegraphics[width=\textwidth]{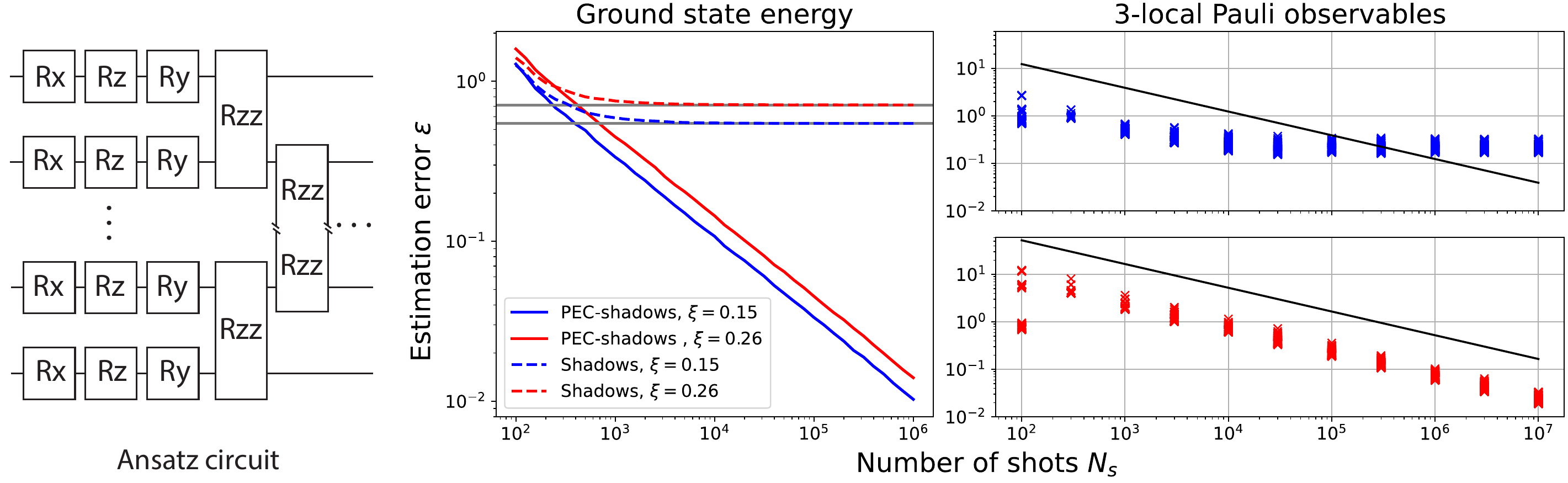}
	\caption{
		(left) A noisy variational Hamiltonian ansatz is used to prepare the ground state of \cref{eq:spin_ring_hamiltonian}
		but our aim is to learn properties of the noise-free state.
		(middle) Energy estimation errors for different noise strengths with conventional shadows (dashed blue, dashed red) 
		and with PEC shadows including readout error mitigation (solid blue, solid red). 
		Bias (grey solid lines) is introduced when the ground state energy $\tr(\rho\mathcal{H})$ is directly estimated
		from the noisy quantum state $\rho$.
		Error mitigated shadows are unbiased as they
		estimate $\tr(\rhoid\mathcal{H})$.
		Increasing the circuit error rate $\xi$ (blue vs. red) increases the bias in standard shadows (dashed blue vs. dashed red)
		and increases the variance of the error mitigated shadows (solid blue vs. solid red). 
		Each data point is an average over $10^4$ experiments of a fixed shot budget $\nshot$.
		(right)
		Error in simultaneously estimating all 3-local Pauli strings without (blue) and with (red) PEC and readout error mitigation
		-- only the 200 observables of the highest estimation error are shown and a circuit error rate $\xi \approx 0.72$ is assumed.
		Errors are significantly below our rigorous bounds from
		\cref{theo:m_properties} (which also take into account the overhead due to readout errors)
		for PEC shadows but the errors for conventional shadows can be above
		their respective bounds from \cite{shadow_huang_2020} due to bias and readout errors (right end of blue).
	}
	\label{fig:error_gs_estim}
\end{figure*}

For example, a typically used symmetry in fermionic simulation is the fermionic particle number parity
which is, however, usually a high-weight operator for standard encodings such as the
Jordan-Wigner encoding. Nevertheless, one can use encodings that come with inherent local
symmetry generators like Majorana loop encodings~\cite{jiangMajoranaLoopStabilizer2019},
or even implement the circuit using some small quantum codes with local stabilisers~\cite{mccleanDecodingQuantumErrors2020}.
However, even if these symmetry \emph{generators}
are local, the number of generators scales with the number of
qubits thus some symmetries they generate are still high-weight.
Hence, in order to efficiently use shadow techniques, we can apply
verification using a constant number of local symmetry generators,
such that the highest-weight symmetry that can be generated is
upper-bounded by some constant. 

We also note that the sampling cost can be reduced when the target
observable $O$ commutes with the symmetry projector
$\Pi_{\mathbb{S}}$ which is often the case in typical
applications. In such a scenario, $\Pi_{\mathbb{S}}O\Pi_{\mathbb{S}} = O\Pi_{\mathbb{S}}$
and thus
\begin{align*}
	\Tr(O\rho_{\textrm{sym}}) = \frac{\Tr(O\Pi_{\mathbb{S}}\rho)}{\Tr(\Pi_{\mathbb{S}} \rho)} =  \frac{\sum_{S \in \mathbb{S}}\Tr(SO\rho)}{\sum_{S \in \mathbb{S}}\Tr(S \rho)}.
\end{align*}
This way the effective observables we need to estimate from shadows
are $SO$ and $S$ for all $S \in \mathbb{S}$ which have a reduced
weight $s{+}q$ compared to the previous $2s{+}q$.


\section{Applications\label{sec:app}}

In this section we showcase how our approach can effectively extend the reach
of noisy quantum computers and explore its practical applications.
	Recall that fully fault-tolerant quantum computers will enable executions of (in principle)
	arbitrarily
	deep circuits thus allowing users to extract expected values via, e.g.,
	amplitude estimation, whose time complexity is superior $O(M/\epsilon)$ but is proportional to the circuit depth.
	In contrast, we focus on application areas where these coherent techniques are prohibitive
	due to circuit depth limitations, e.g., 
	due to non-negligible logical error rates expected in early fault-tolerant devices.
	The advantage of classical shadows is that they only require an increase of circuit depth that is
	independent of the state preparation circuit -- and this 
	increase is negligible for Pauli shadows. Since
	the present approach has a sample complexity $\mathcal{O}( \log(M) / \epsilon^2 )$ it is particularly well suited for applications
	where the aim is to extract a large number $M$ of properties.

For instance, noisy quantum computers in either the late NISQ era or in the early fault-tolerance era
will enable us to simulate the time evolution of quantum states or to prepare ground or eigenstates~\cite{cerezoVariationalQuantumAlgorithms2021a, endoHybridQuantumClassicalAlgorithms2021, bharti2021noisy, covar,shadow_spec_2023,koczor2020quantumAnalytic, PhysRevA.106.062416}.
Our approach can then be used to accurately and efficiently extract a large number of properties of these states
provided that the noise rates are reasonable, i.e., the overhead $\lVert g \rVert_1$ is moderate.
In these application areas a fixed precision, such as chemical accuracy $\epsilon \approx 10^{-3}$,
	is often sought~\cite{mcardle2020quantum,cerezoVariationalQuantumAlgorithms2021a, endoHybridQuantumClassicalAlgorithms2021, bharti2021noisy}.

\subsection{Ground-state preparation \label{sec:eigenstate_prep}}
We first consider a spin-ring Hamiltonian as
\begin{equation}\label{eq:spin_ring_hamiltonian}
    \mathcal{H} = \sum_{k \in \text{ring}(N)} \omega_k Z_k + J \Vec{\sigma}_k\cdot\Vec{\sigma}_{k+1},
\end{equation}
with coupling $J = 0.3$, on-site interaction strengths uniformly randomly
generated in the range $-1 \leq \omega_k \leq 1 $ and $\vec{\sigma}_k = (\sigma^x_k, \sigma^y_k, \sigma^z_k)^T$
is a vector of single-qubit Pauli matrices.
This spin problem is relevant in condensed-matter physics in understanding many-body localisation \cite{nandkishore_2015}
but is challenging to simulate classically for large $N$~\cite{childs_2018, luitz_2015}. 
A broad range of techniques are available in the literature for finding eigenstates of such quantum Hamiltonians using near-term
or early fault-tolerant quantum computers~\cite{cerezoVariationalQuantumAlgorithms2021a, endoHybridQuantumClassicalAlgorithms2021, bharti2021noisy, covar}.
Here we prepare the ground state of this model using a variational Hamiltonian
ansatz in \cref{fig:error_gs_estim} (left) of $l= 5$ layers on $12$ qubits
and, as we detail in  \cref{app:numerics}, we assume
a biased Pauli noise model that can be learned efficiently using techniques from~\cite{berg2022probabilistic}
while assuming a readout error model from \cref{sec:readout}.

\noindent\textbf{Ground state energy with PEC:}
\cref{fig:error_gs_estim} (middle) shows the error in the ground state energy estimated using conventional shadows (dashed blue and dashed red)
and PEC shadows (solid blue and solid red) for an increasing number of shots $\nshot$. 

Our ansatz  circuit would ideally prepare the ground state $\rhoid$ but due to gate noise
we actually prepare the noisy state $\rho$. Thus, conventional shadows (dashed blue, dashed red) converge to 
a plateau corresponding to the biased energy $\tr({\rho\mathcal{H}})$ (solid grey).
This bias is significantly increased as we increase the circuit error rate from $\xi \approx  0.15$ to $\xi \approx 0.26$ (dashed blue vs. dashed red),
which is the expected number $\xi = \sum_k p_k$ of errors in the full circuit as explained in \cref{sec:pec}.
In contrast, PEC shadows that include measurement-error mitigation (solid blue and solid red) 
converge to the true energy $\tr({\rhoid\mathcal{H}})$ in standard shot-noise scaling
$O(1/\sqrt{\nshot})$.

\noindent\textbf{Local properties with PEC:}
Besides Hamiltonian energy estimation, which is one of the typical subroutines in quantum computing,
there is also significant value in simultaneously determining many local observables' expectation values.
For example, the rich information from classical shadows can be used to significantly
improve parameter training or to directly estimate Hamiltonian energy gaps
through the use of efficient classical post-processing~\cite{covar, shadow_spec_2023}.
In \cref{fig:error_gs_estim}(right), we plot errors when simultaneously estimating all 3-local Pauli operators
for an increasing number of shots $\nshot$.
\cref{fig:error_gs_estim}(red) shows that the errors in PEC shadows are always significantly
below the theoretical bounds (black line) from \cref{theo:m_properties} confirming looseness of the bounds
(assuming success probability $\delta = 10^{-3}$, and $M = 3^3{12 \choose 3}= 5940$).
\cref{fig:error_gs_estim}(blue) shows the errors in conventional shadows are below their bounds (with $\lVert g\rVert_1=1$)
only for a small number of shots but then asymptotically reach a plateau due to circuit noise.

\begin{figure}[tb]
	\centering
	\includegraphics[width=0.48\textwidth]{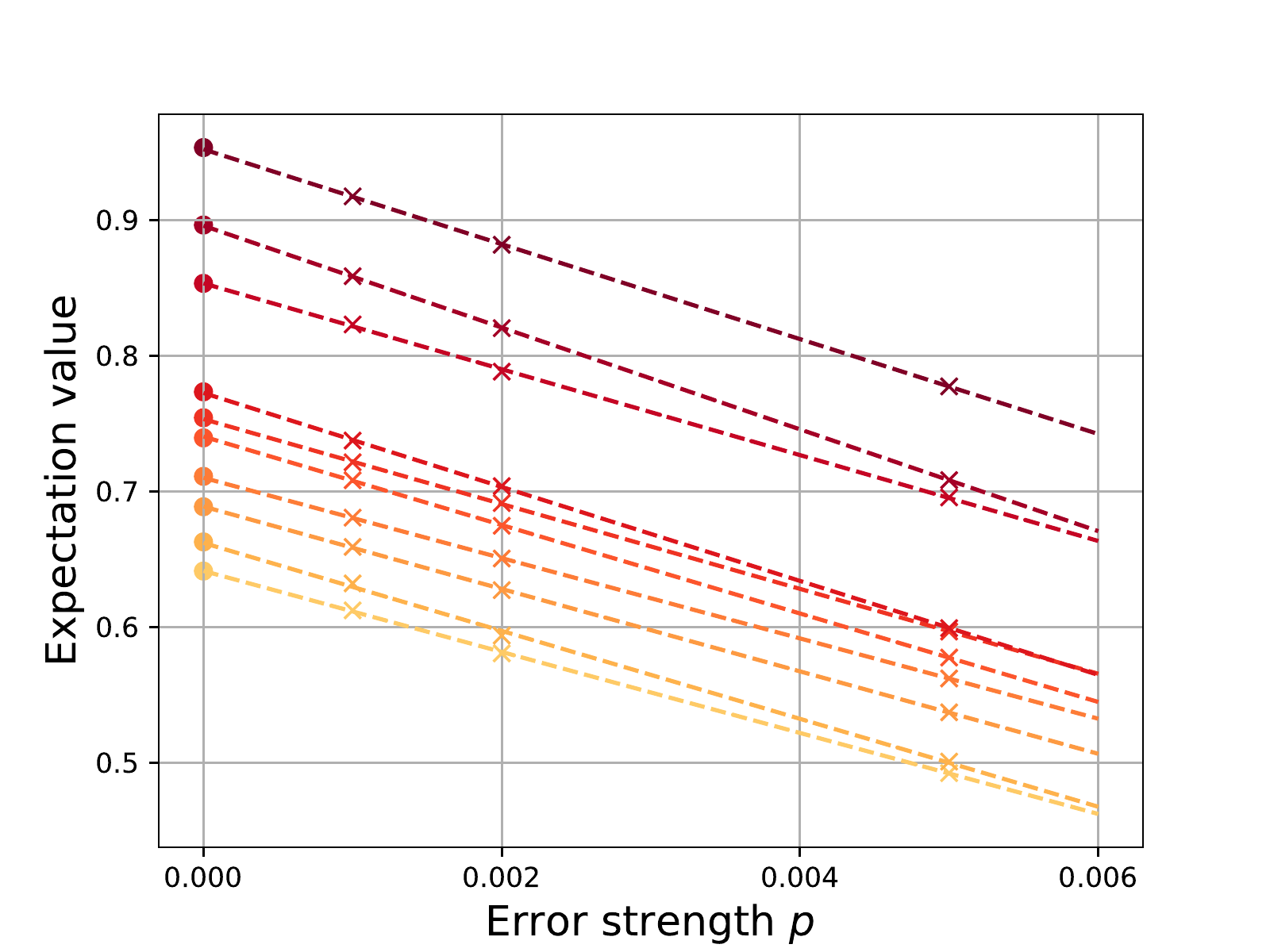}
	\caption{Simultaneously estimating all 3-local Pauli operators using error extrapolated shadows.
		We estimated noisy expected values (crosses) from shadows of size $N_s = 10^7$ by increasing the native depolarising error rate
		$p= 10^{-3}$ to higher levels $\left\{2\times10^{-3}, 5\times10^{-3}\right\}$
		by randomly sampling noisy circuit variants.
		Using a linear model function we then extrapolate to zero noise to obtain an error mitigated expectation value close to the ideal ones (disks).
	}
	\label{fig:extrapol_shadows}
\end{figure}

\begin{figure*}
	\centering
	\includegraphics[width=\textwidth]{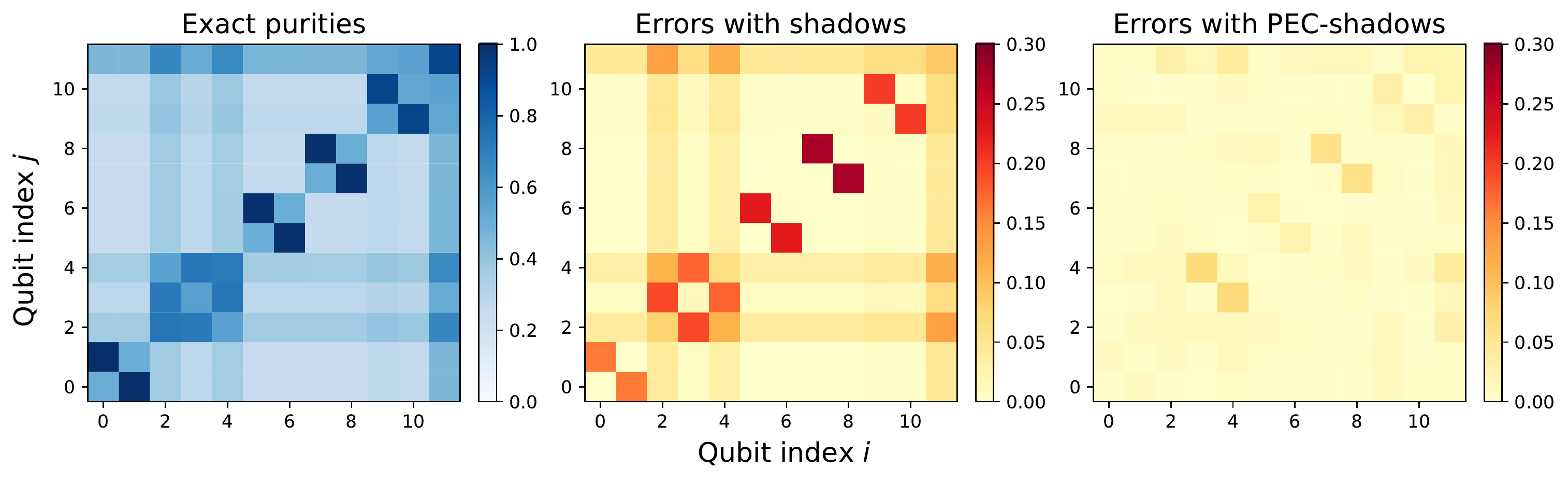}\hfill
	\caption{A noisy variational Hamiltonian ansatz is used to prepare the ground state of \cref{eq:heinseberg_chain_hamil} whose ideal, noise-free Rényi entropies $R_{Q}$ we can learn with PEC shadows. We plot purities $\tr(\rho_{Q}^2)$ as a proxy for $R_Q := -\log \tr\left(\rho_{Q}^2\right)$. (left) Purity heat map in the noiseless case and infinite shot limit.
    An increasing value indicates that the subsystem $Q$ is less entangled with the remaining qubits.
    (middle-right) Absolute error in the purities due to gate noise for a circuit error rate $\xi = 0.6$
    and due to finite repetition using $N_s = 10^5$.
    (middle) Although the entanglement pattern is approximately recovered with conventional shadows,
    in some instances we observe substantial errors, i.e., the largest error is $0.27$.
    (right) Absolute errors with PEC shadows are significantly smaller, i.e., the largest error is
    $ 7\times 10^{-2}$ but this figure could be further reduced by increasing $N_s$.}
	\label{fig:entropy}
\end{figure*}

\noindent\textbf{Local properties with extrapolation:}
We now consider the same task of simultaneously estimating expectation values of Pauli operators but we use error extrapolation.
Here we start by generating  shadows $S(\rho_{p_1}, N_s), ..., S(\rho_{p_n}, N_s)$ at different noise strengths
that we use to compute the noisy Pauli expectation values.
\cref{fig:extrapol_shadows} shows 10 examples of expected values (crosses) as a function of noise strength
and the respective linear models we fit (dashed lines). The intercept of the fitted model (dashed lines)
is our estimate of the exact expected value (disks) and is indeed reasonably close in the example.
While ZNE has been very effective and typically has a lower measurement overhead then PEC, it is generally
biased.

\subsection{Error mitigated estimation of entanglement entropies\label{sec:renyi}}
Finally, we consider an application for which classical shadows are a primary enabler
but for which error mitigation techniques have been less explored~\cite{cai2022quantum}.
As opposed to studying entanglement properties or verifying the presence thereof
in mixed quantum states~\cite{estimate_measure_guehne_2007,gme_20_blatt_2018,bounding_entanglement_measures_guehne_2019,machine_learning_quantum_many_body_preskill_2022}, here
our primary goal is to extend the reach of noisy quantum computers:
we aim to study entanglement properties of ideally pure states
which are prepared by quantum algorithms, such as phase estimation or VQE.
For example, near-term quantum computers will enable us to prepare eigenstates~\cite{cerezoVariationalQuantumAlgorithms2021a, endoHybridQuantumClassicalAlgorithms2021, bharti2021noisy, covar}
of quantum Hamiltonians and error mitigated entanglement measures can be used for, e.g.,
characterising phase transitions.
Similarly, one could simulate the
time evolution of a collision of two molecules with an early fault-tolerant quantum computer
and investigate how entanglement builds up across the individual subsystems. Furthermore, efficiently characterising
many local correlations in a state can be used to train DFT models for accurate classical simulations~\cite{baker2020density}.

We consider the Heisenberg chain
\begin{equation}\label{eq:heinseberg_chain_hamil}
    \mathcal{H} = \sum_{k} J_k \Vec{\sigma}_k\cdot\Vec{\sigma}_{k+1},
\end{equation}
with uniformly random  $-1 \leq J_k \leq 1$
and prepare its ground state with a variational Hamiltonian ansatz of $l = 8$ layers on $12$ qubits.
This system was used in ref.~\cite{shadow_huang_2020} to illustrate the power of classical shadows
in predicting entanglement entropies. However, the ground state was approximated by a set of noise-free
singlet states~\cite{dasgupta_1980, ma_1979} whereas we assume a noisy quantum computer is used for state preparation.

We use PEC shadows to extract purities $\tr(\rho_{Q}^2)$ for all single and two-qubit subsystems $Q$;
These purities then define Rényi entropies as $R_{Q} := -\log\tr(\rho_{Q}^2)$. 
In \cref{fig:entropy}, we plot the exact purities in the noiseless case -- 
disjoint blocks involving two qubits confirm that the ground
state could be approximated by a tensor product of noise-free singlet states. 

\cref{fig:entropy} (middle) shows the errors in estimating local purities using shadows of size
$N_s = 10^5$ for a circuit error rate $\xi = 0.6$. Even for this moderate error rate
conventional shadows are significantly impacted by imperfections and result in errors as large
as $0.27$ -- whereas for an increasing noise rate all purities converge to a constant value of $1/d$
where $d$ is the subsystem dimension.
In contrast, PEC shadows drastically improve the accuracy in \cref{fig:entropy} (right)
and the largest error is approximately $7\times10^{-2}$ at a number of samples $\nshot = 10^5$. 

\subsection{Further applications}
The techniques presented in this work enable us to approximate an unbiased estimator
of an ideal noise free state $\rhoid$
which can be enabling for a broad range of further practical applications  
that we defer to follow up works. For example, ref.~\cite{shadow_huang_2020}
proposed that classical shadows with randomised Clifford measurements can be used
to predict fidelities, such as the fidelity of $\rho$ with respect to a known
state $\psi$. One can imagine applications where the fidelity $\langle \psi |\rho| \psi \rangle$
is not a relevant indicator due to the impact of noise on $\rho$ and one rather aims to
predict $\langle \psi |\rhoid| \psi \rangle$, e.g., to quantify how well a variational
quantum circuit or phase estimation can prepare a known ground state thereby verifying
a circuit structure under the presence of gate noise.

Furthermore, the quantum Fisher information (QFI), which is a key quantity in quantum metrology,
can be bounded and approximated using classical shadows via techniques of Ref.~\cite{qfisher_bound_rath_2021}.
Indeed, in certain applications the relevant quantity might not be the QFI of the noisy state $\rho$
but rather the QFI of the noise-free state $\rhoid$ which can be approximated with our techniques~\cite{koczor2020variational}.

We also finally note possible ``reverse'' applications where classical shadows
can be used to improve error mitigation techniques; an obvious one is perhaps the use of
shadow tomography in explicitly reconstructing the noisy gate channels from \cref{sec:pec} -- typically one assumes
the gate channels are local and thus the approach is efficient.
In contrast, in learning-based error mitigation one does not reconstruct the noisy gate channels
but rather aims to directly learn the quasiprobabilities $g_{\underline{k}}$ by
running classically simulable circuits on a noisy quantum computer and comparing observable
expectation values to classically simulated
ones~\cite{strikis2021learning,PhysRevResearch.3.033098,montanaro2021error, cai2022quantum}.
While one can exclusively train a model for one specific observable with conventional measurement schemes, classical shadows allow for simultaneously estimating a large number of expected values.
We can thereby efficiently train error models that mitigate the impact of errors
in all local operator measurements.
The approach might similarly be useful in measuring many Pauli operators
in case of learning sparse Pauli models~\cite{berg2022probabilistic}

\section{Discussion and Conclusion\label{sec:concl}}
In this work we consider the powerful classical shadows methodology which
allow us to obtain an efficient classical representation of a quantum state $\rho$
and thus to simultaneously predict many of its properties in classical post-processing.
A major difficulty concerning near-term and early fault-tolerant quantum computers 
is that they can only prepare noisy quantum states $\rho$ from
which we would estimate corrupted properties; This challenge
motivated the field to develop quantum error mitigation techniques
that allow us to estimate expected values $\tr[O \rhoid]$  of observables $O$
in an ideal noise-free state $\rhoid$ but with having access only to noisy expected values.

We consider a range of typical quantum error mitigation techniques and
generalise them from single expected-value measurements
to the case of mitigating errors in classical shadows.
We find that Probabilistic Error Cancellation
is the most well-suited candidate which motivates
us to develop a thorough theory of PEC shadows.
In the conventional PEC approach one learns error characteristics of
the device and counters them by a probabilistic implementation of the inverse noise channel --
thus the only source of noise is
due to a possibly imperfect knowledge of
gate-error characteristics and due to finite circuit repetition.
Under the assumption that the error model of the quantum device has been appropriately
learned such that a quasiprobability representation is known, we prove that PEC shadows
are an unbiased estimator of the ideal state $\rhoid$.
We additionally prove the following rigorous performance guarantees.

First, we prove  bounds on the number of samples required to simultaneously predict
many linear properties of the ideal quantum state $\rhoid$. 
Second, the fact that we use noisy quantum circuits to predict ideal properties manifests
in a multiplicative measurement overhead -- this overhead is identical to the cost
of the conventional PEC approach and grows exponentially with the number of noisy gates.
Third, we prove rigorous sample complexities for predciting non-linear properties of the ideal
states. 

We note that our results are completely general and apply to any shadow ensemble $E$
via \cref{eq:idealised_measurements} and to any linear or non-linear property of the quantum state.
Furthermore, we provide practical post-processing protocols for the pivotal scenario of randomised
measurements in Pauli bases.
Finally, we demonstrate in numerical simulations the usefulness of PEC shadows and error extrapolated shadows,
and conclude that these techniques may be instrumental in practical applications of near-term and early fault-tolerant machines.

We note that previous works have already explored applying error mitigation techniques to
classical shadows focusing on errors in the POVMs $E$~\cite{koh_classical_2022, chen_2021} assuming the aim is to estimate
shadows of an input state $\rho$. Furthermore, classical shadows of $\rho$ have been 
used to classically estimate expected values $\tr[O \rho^n]/\tr[\rho^n]$ thereby classically performing
Error Suppression by Derangements (ESD)~\cite{PhysRevX.11.031057}, Virtual Distillation (VD)~\cite{huggins2020virtual}.
This ultimately allows us to estimate expected values in the dominant eigenvector of $\rho$
which is an approximation~\cite{koczor2021dominant} to $\rhoid$ but at an exponential 
cost in the number of qubits and in the number of noisy gates. In contrast, the techniques we present are efficient in the
sense that the sample complexity does not directly depend on the number of qubits
but rather depends exponentially on the number of noisy gates, here via the norm $\lVert g \rVert_1^2$.
As with usual error mitigation techniques, the approach is limited to a number of gates
$\nu \in \mathcal{O}(p^{-1})$ given by the inverse of the per-gate error rate --
beyond this threshold the effect of errors escalates exponentially~\cite{quek_2023, takagi_fundamental_2022,takagi2022universal}.
State-of-the-art theoretical lower bounds suggest a slightly more optimistic picture
whereby circuits of $\mathrm{poly} \log \log N$ depth provably yield exponential decrease of fidelity
as opposed to constant depth suggested by $\lVert g \rVert_1^2$.

In summary, the present work leverages an existing, rich toolbox of quantum error mitigation ideas
and generalises powerful classical shadows to the pivotal scenario of approximating properties
of an ideal quantum state $\rhoid$. As we demonstrate in a broad range of
examples, these quantum error mitigated
classical shadows are very intuitive, easy to use in practice and may play a central role in exploiting 
near-term and early fault-tolerant quantum computers.
We discuss a broad range of further possible use cases
and anticipate the present work will stimulate further advancements in the field.

\section*{Acknowledgments}
The authors thank Simon Benjamin, Dan Browne, Otfried Gühne, Ryuji Takagi, and Zoltán Zimborás for helpful comments
on drafts of the manuscript.
B.K. thanks the University of Oxford for
a Glasstone Research Fellowship and Lady Margaret Hall, Oxford for a Research Fellowship.
J.S. and H.C.N are supported by the Deutsche Forschungsgemeinschaft (DFG, German Research Foundation, project numbers 447948357 and 440958198), the Sino-German Center for Research Promotion (Project M-0294), the ERC (Consolidator Grant 683107/TempoQ), and the German Ministry of Education and Research (Project QuKuK, BMBF Grant No. 16KIS1618K).
J.S. also acknowledges the support from the House of Young Talents of the University of Siegen.
Z.C. is supported by the Junior Research Fellowship from St John’s College, Oxford.
The numerical modelling involved in this study made
use of the Quantum Exact Simulation Toolkit (QuEST), and the recent development
QuESTlink~\cite{QuESTlink} which permits the user to use Mathematica as the
integrated front end, and pyQuEST~\cite{pyquest} which allows access to QuEST from Python.
We are grateful to those who have contributed
to all of these valuable tools. 
The authors would like to acknowledge the use of the University of Oxford Advanced Research Computing (ARC)
facility~\cite{oxford_arc} in carrying out this work
and specifically the facilities made available from the EPSRC QCS Hub grant (agreement No. EP/T001062/1).
The authors also acknowledge funding from the
EPSRC projects Robust and Reliable Quantum Computing (RoaRQ, EP/W032635/1)
and Software Enabling Early Quantum Advantage (SEEQA, EP/Y004655/1).

\appendix

\section{Details of PEC Shadows}

\subsection{Unbiased estimators (\texorpdfstring{\cref{lemma_PEC_unbiased} and \cref{theo:mitigated_shadow}}))}
\label{app:unbiasedness}

\begin{proof}[Proof of \cref{lemma_PEC_unbiased}]
	The statement directly follows from the fact that $\vert \vert g \vert \vert_{1} \mathrm{sign}(g_{\underline{k}}) \mathcal{G}_{\underline{k}}$
	is an unbiased estimator for the ideal operation $\mathcal{U}_{\text{circ}}$.
	In particular, as we sample $\underline{k}$ according to the probability distribution
	$p(\underline{k})$, we obtain the expectation as
	\begin{align*}
		\underset{\underline{k}}{\mathbb{E}} \,[\hat{\rho}]
		& =  \underset{\underline{k}}{\mathbb{E}} [ \lVert g\lVert_1 \, \mathrm{sign} (g_{\underline{k}}) \mathcal{G}_{\underline{k}} (\vert 0 \rangle \langle 0 \vert )] \\
		& =  \sum_{\underline{k} } p(\underline{k}) \lVert g\lVert_1 \, \mathrm{sign} (g_{\underline{k}}) \mathcal{G}_{\underline{k}} (\vert 0 \rangle \langle 0 \vert ). \\
	\end{align*}
The above expression can be simplified collecting the constant factors as
$p(\underline{k}) \vert \vert g \vert \vert_{1} \mathrm{sign}(g_{\underline{k}}) = \mathrm{sign}(g_{\underline{k}}) \vert g_{\underline{k}} \vert =  g_{\underline{k}}$ and thus we obtain
the quasiprobability decomposition
\begin{equation*}
	\underset{\underline{k}}{\mathbb{E}} \,[\hat{\rho}]
	 = \sum_{\underline{k} } g_{\underline{k}} \mathcal{G}_{\underline{k}}
	( \vert 0 \rangle \langle 0 \vert)
	 = \mathcal{U}_{circ} (\vert 0 \rangle \langle 0 \vert) = \rho_{id}. 
\end{equation*}
\end{proof}

\begin{proof}[Proof of \cref{theo:mitigated_shadow}]
Using the abbreviation $\hat{\rho}_{id} \equiv \hat{\rho}_{\underline{k} ,l}$
we calculate the expected value as
\begin{equation}\label{eq:app_expectation_ideal_state}
    \underset{\underline{k},l}{\mathbb{E}} [\hat{\rho}_{id}]
    =
    \sum_{\underline{k} ,l} p_{\underline{k}} \, q_l  \, \hat{\rho}_{\underline{k} ,l},
\end{equation}
where $p_{\underline{k}} = \vert  g_{\underline{k}} \vert  / \lVert g \rVert_1$ is the probability of choosing the circuit variant
$\mathcal{G}_{\underline{k}}$ from \cref{def:decompose}
and we also use the probability $q_l = \tr[\mathcal{G}_{\underline{k}} (\vert 0 \rangle \langle 0 \vert) E_{l}]$ 
of observing the POVM outcome $l$.
We obtain the expected value by substituting these in Eq.~\eqref{eq:app_expectation_ideal_state} as
	\begin{align*}
			\underset{\underline{k},l}{\mathbb{E}} [\hat{\rho}_{id}] &= 
			\sum_{\underline{k} ,l} \frac{\vert g_{\underline{k}} \vert}{  \lVert g \rVert_1 } \, 
			\tr[\mathcal{G}_{\underline{k}} (\vert 0 \rangle \langle 0 \vert) E_{l}] \, \lVert g \rVert_1\, \sign (g_{\underline{k}}) C_{E}^{-1} (E_{l}).
	\end{align*}
Here we can collect and simplify all constant factors as $\lVert g \rVert_1 \tfrac{\vert g_{\underline{k}} \vert}{  \lVert g \rVert_1 }  \sign (g_{\underline{k}}) = g_{\underline{k}}$
and simplify the expected value as
\begin{align}
	\underset{\underline{k},l}{\mathbb{E}} [\hat{\rho}_{id}] &= 
	 \sum_{l} \tr\Big[\big(\sum_{\underline{k}} g_{\underline{k}} \mathcal{G}_{\underline{k}}\big)
	\big(\vert 0 \rangle \langle 0 \vert\big) E_{l}\Big] C_{E}^{-1} (E_{l}) \\
	&= \sum_{l} \tr[\rho_{id} E_{l}] C_{E}^{-1} (E_{l}) \\
	&= C_{E}^{-1} \sum_{l} \tr[\rho_{id} E_{l}]  E_{l} \\
	&= (C_{E}^{-1} C_{E}) (\rho_{id}) = \rho_{id}.
\end{align}
Above in the first equality we simply used the linearity of the trace operation
while in the second equality we used that by definition
$\sum_{\underline{k}} g_{\underline{k}} \mathcal{G}_{\underline{k}} \vert 0 \rangle \langle 0 \vert = \rhoid$.
We finally substituted the definition of
the measurement channel $C_{E} (\cdot)$ given by Eq.~\eqref{eq:measurement_channel}.
\end{proof}

\subsection{Shadow norms \label{app:shadow_norm}}

Before proving \cref{lemma:expval_variance}, let us define the
shadow norm of an operator $O$ and calculate it explicitly for
the practically important scenario when $O$ is a local Pauli string.
\begin{lemma}\label{lemma:shadow_norm}
	We define the shadow norm with
	respect to the generalized measurement $E$	as
	\begin{align}
		\lVert O \rVert_{E}^{2} := \lVert \sum_{l=1}^{\nes} \tr \, [\hat{\rho}_{l} O]^{2} E_{l} \rVert_\infty,
	\end{align}
	where $\lVert \cdot \rVert_\infty$ denotes the maximal eigenvalue of the corresponding operator
	and $\hat{\rho}_{l} = C_{E}^{-1}(E_{l})$. For the specific case of
	Pauli-basis measurements and observables that are $q$-local Pauli strings, the
	squared shadow norm is given as $3^{q}$. 
\end{lemma}
\begin{proof}
When formulating shadow tomography with generalized measurements,
the case of uniformly sampled Pauli-basis measurements corresponds to the so
called octahedorn POVM~\cite{povm_shadow_nguyen_2022}, where the effects on a single qubit are given by
\begin{equation*}
    E_{j} =\tfrac{1}{3} Q^\dagger_j | b \rangle \langle b | Q_j,
\end{equation*}
where $b \in \{0,1\}$ is a single bit and $Q_j$ is one of the three basis transformation unitaries that allow us to
measure in the bases of the Pauli $X$, $Y$ and $Z$ operators.

Thus the effect is equivalent to $ \tfrac{1}{3} \vert t^{\pm} \rangle \langle t^{\pm} \vert$ for $t \in \lbrace x,y,z \rbrace$
where  $\vert t^{\pm} \rangle$ denotes the eigenvector corresponding to eigenvalue $\pm1$ of the single-qubit
Pauli-$t$ operator. It follows from the symmetry of the measurement~\cite{povm_shadow_nguyen_2022}
that the shadows can be computed directly from the effects as
\begin{equation}\label{eq:snapshot_analytically}
    \hat{\rho}_{l} = 9E_{l} - \openone.
\end{equation}

For the case of a system consisting of $n$ qubits where one aims to estimate local observables of the form $O = O_{1} \otimes \cdots \otimes O_{n}$ and the measurement is given by the tensor product of local measurements $E_{j_{1}}^{(1)} \otimes \cdots \otimes E_{j_{n}}^{(n)}$ with $E^{(j)}$ denotes the POVM acting on the $j$th qubit, the shadow norm similarly factorizes as 
$\lVert O \rVert_{E}^{2} = \prod_{j} \, \lVert O_{j} \rVert_{E^{(j)}}^{2}$.

We now consider the case when the single-qubit operator $O_{j}$ acting on the $j$th qubit is a Pauli operator $X$, $Y$ or $Z$
and thus $\tr[ O_j]=0$. By the previous discussion, it is sufficient to only consider a single qubit, thus we will suppress the index $j$. This yields the shadow norm 
\begin{align}
   \lVert O \rVert_{E}^{2}
    & = \lVert  \sum_{l=1}^{6} \text{Tr}[\hat{\rho}_{l} O]^{2} E_{l} \rVert_{\infty} \\
    &= \lVert \sum_{t^{\pm}} \frac{1}{3} \text{Tr}[(3 \vert t^{\pm} \rangle \langle t^{\pm} \vert) O]^{2} \vert t^{\pm} \rangle \langle t^{\pm} \vert \, \rVert_{\infty}\\
    &= 3 \lVert \sum_{t^{\pm}} \langle t^{\pm} \vert O \vert t^{\pm} \rangle^{2} \vert t^{\pm} \rangle \langle t^{\pm} \vert \, \rVert_{\infty} \label{eq:proof_pauli_norm}
\end{align}
Now observe that if $O,T\in \lbrace X,Y,Z \rbrace$ with $\vert t^{\pm} \rangle$ are the
normalized eigenvectors of $T$ to eigenvalues $\pm1$,  we have due to the anticommutation
relation  $ \delta_{O,T}  = \frac{1}{2} \langle t^{\pm} \vert \lbrace O,T \rbrace \vert t^{\pm} \rangle
= \pm  \langle t^{\pm} \vert O \vert t^{\pm} \rangle$.
This implies that the sum in Eq.~\eqref{eq:proof_pauli_norm} collapses to the
identity $\openone$. Hence we obtain $\lVert O \rVert_{E}^{2} = 3$.
When the single-qubit observable is the identity $O= \openone$ we obtain a shadow
norm $\lVert O \rVert_{E}^{2} = 1$. Consequently, for $q$-local Pauli strings acting on $n$ qubits the squared shadow norm is $\lVert O \rVert_{E}^{2} = 3^q$, thus independent of $n$.
We explain in \cref{app:numerics} how this bound is modified when considering the effect of readout errors.
\end{proof}

In practice it is often the case that the set of targeted observables posses a certain structure.
If this is the case, small variations to classical shadow protocol in which the measurement basis
is sampled uniformly at random can yield a substantial improvement with respect to sample complexity~\cite{locally_biased_shadows_hadfield_2022}.

For instance, in electronic structure problems 
where one aims to, e.g., determine the ground-state of molecules using a quantum algorithm,
one typically starts by transforming the molecular Hamiltonian into a qubit Hamiltonian
as a sum of Pauli observables by means of an appropriate \textit{mapping}.
Common types of such mappings are Jordan-Wigner (JW), Bravyi-Kitaev (BK) and the parity (P)
transformation~\cite{mcardle2020quantum}. Here it is important to note that depending on the encoding, the
different Pauli operators $X,Y,Z$ appear with different frequencies in the corresponding
qubit observable. For instance in BK encoding, the appearance of Pauli-$Y$ operators is
suppressed compared to $X$ and $Z$. Consequently, measuring the different Pauli
basis uniformly on each qubit, i.e., using the octahedron measurement, would be very wasteful.

A similar statement concerning sample complexity as in Lemma~\ref{lemma:shadow_norm} can
be made for the case of locally biased shadows~\cite{locally_biased_shadows_hadfield_2022, huang_derandomization_2021}.
Let assume that the bias is $p_{x},p_{y},p_{z}$ where $p_{t}$ is the probability for performing the measurement in Pauli $t$ basis. The corresponding POVM would be $E_{t^{\pm}} = p_{t} \vert t^{\pm} \rangle \langle t^{\pm} \vert $. 
Then the classical shadow based on measurement outcome would be 
\begin{align}
    \hat{\rho}_{t^{\pm}} = p_{t}^{-2} E_{t^{\pm}} - \frac{\mu - p_{t}^{2}}{2p_{t} \mu} \openone,
\end{align}
where $\mu = p_{x}^{2} + p_{y}^{2} + p_{z}^{2}$. With this, given a Pauli string, one can directly calculate the shadow norm. 

\section{Proofs of performance guarantees}
\subsection{Variance of linear properties
(\texorpdfstring{\cref{lemma:expval_variance}}))
	}
\label{app:variance}
\begin{proof}[Proof of \cref{lemma:expval_variance}]
	Note that $ \text{Var}[\hat{o}] = \mathbb{E}[(\hat{o} - \mathbb{E}[\hat{o}])^{2}] $.
	As $\hat{\rho}_{id}$ is unbiased, we have $\mathbb{E}[\tr(O \hat{\rho}_{id})]^{2} = \langle O \rangle^{2}$ and thus $\text{Var}[\hat{o}] = \mathbb{E}[\tr(O \hat{\rho}_{id})]^{2} - \langle O \rangle^{2} \geq \mathbb{E}[\tr(O \hat{\rho}_{id})]^{2} $.
	Hence it remains to bound the term
	\begin{align*}
			\underset{\underline{k},l}{\mathbb{E}} \Big[ \tr(O \hat{\rho}_{id})^{2} \Big] 
			&=
			\underset{\underline{k},l}{\mathbb{E}}\Big[ \, 
			\text{Tr}\big[   O \lVert g \rVert_1 \sign (g_{\underline{k}}) C_{E}^{-1}(E_{l})    \big ]^2  
			\, \Big] \\
			&=
			\lVert g \rVert_1^{2} \, \underset{\underline{k},l}{\mathbb{E}} \Big[ \tr \big[   O C_{E}^{-1} (E_{l})  \big]^{2}  \Big].
	\end{align*}
We can now calculate the expectation by recalling that $p_{\underline{k}} = \vert g_{\underline{k}} \vert / \lVert g \rVert_1$
is the probability from \cref{def:decompose}
of choosing the circuit variant $\mathcal{G}_{\underline{k}}$
and $q_l = \tr[\mathcal{G}_{\underline{k}} (\vert 0 \rangle \langle 0 \vert) E_{l}]$ 
is the probability of observing the POVM outcome $l$.
Thus the above expectation is calculated as
\begin{align*}
	&\lVert g \rVert_1^{2} \sum_{\underline{k},l} \frac{\vert g_{\underline{k}} \vert }{  \lVert g \rVert_1  }  
	\times
	\tr \big[  \mathcal{G}_{\underline{k}} (\vert 0 \rangle \langle 0 \vert) E_{l}  \big] \times \tr\big[   O C_{E}^{-1}(E_{l})  \big]^{2} \\
	&=
	\lVert g \rVert_1^2 \, \sum_{l} \tr \big[  \Omega(\vert 0 \rangle \langle 0 \vert) E_{l} \big] \times \tr \big[O C_{E}^{-1} (E_{l}) \big ]^{2}.
\end{align*}
Above we introduced	$\Omega :=  \lVert g \rVert_1^{-1} \, \sum_{\underline{k}} \vert g_{\underline{k}} \vert \mathcal{G}_{\underline{k}}$
	which is actually a permissible quantum channel~\cite{nielsen_chuang_2010}, i.e.,
	a CPTP map since by definition it is a convex combination of CPTP maps $\mathcal{G}_{\underline{k}}$.
	This expression is similar to the one in Ref.~\cite{shadow_huang_2020}.
	
	The above expression can be upper bounded by replacing the initial state $\vert 0 \rangle \langle 0 \vert$
	by a maximization over all states $\sigma$. Thus we obtain the upper bound
	\begin{align} \label{eq:3}
		\begin{split}
			&\underset{\underline{k},l}{\mathbb{E}}\Big[ \tr(O \hat{\rho}_{id})^{2}  \Big] \\
			& \leq 
			\lVert g \rVert_1^2 \, \underset{\sigma}{\text{max}} \, \sum_{l} \tr \big[  \Omega(\sigma) E_{l} \big] 
			\times  \tr \big[  O C_{E}^{-1}(E_{l}) \big]^{2} \\
			& = 
			\lVert g \rVert_1^2 \, \underset{\sigma}{\text{max}} \,
			\tr  \Big[  
			\Omega(\sigma) \sum_{l} \Big( \tr\big[ O \hat{\rho}_{l} \big]^{2} E_{l} \Big)
			\Big], 
		\end{split}
	\end{align}
where we moved the summation inside the trace. By introducing the abbreviation $\Gamma = \sum_{l} \text{Tr}[O \hat{\rho}_{l}]^{2} E_{l} $
we obtain the upper bound as
\begin{align*}
	\underset{\underline{k},l}{\mathbb{E}}\Big[ \tr(O \hat{\rho}_{id})^{2}  \Big]
	&\leq \lVert g \rVert_1^{2} \, \underset{\sigma}{\text{max}} \, \tr \Big[  \Omega(\sigma) \Gamma  \Big] \\
	& \leq 
	\lVert g \rVert_1^{2} \, \underset{\sigma}{\text{max}} \, \tr\big[ \sigma \Gamma \big] \\
	&= \lVert g \rVert_1^{2} \, \lVert \Gamma \rVert_\infty \\
	&= \lVert g \rVert_1^{2} \, \Vert O \rVert_{E}^{2}.
\end{align*}
	Above we used that $\Omega(\sigma)$ is a valid density matrix 
	and thus upper bounded the trace via the operator norm $\tr\big[ \sigma \Gamma \big] \leq \lVert \Gamma \rVert_\infty$
	as the largest singular value of $\Gamma$, which is by definition the shadow norm
	from \cref{lemma:shadow_norm}. Since $\langle O \rangle^{2} \geq 0$,
	we obtain $\var[\hat{o}] \leq \lVert g \rVert_1^{2} \, \lVert O \rVert_{E}^{2} - \langle O \rangle^{2}
	\leq \lVert g \rVert_1^{2} \, \lVert O \rVert_{E}^{2}$.
\end{proof}

Let us make a few observations. (a) recall that conventional classical shadows
make no assumption about the input state $\rho$~\cite{shadow_huang_2020}. In contrast, in our case a circuit description
$\mathcal{U}_{\text{circ}} \vert 0 \rangle \langle 0 \vert$ of the ``input state'' $\rho_{id}$
is actually part of the protocol. 
Of course, knowing such a description of the input state does not allow one to
classically efficiently predict its properties without using classical shadows unless
the circuit $\mathcal{U}_{\text{circ}}$ has some special properties permitting
efficient classical simulation, such as Clifford circuits.

(b) the proof in \cref{lemma:expval_variance} involves a maximization over density matrices
such that our bounds are independent of the particular quasiprobability decomposition  
and thus depends only on the norm $\lVert g \rVert_1^{2}$.
(c) it can be expected that the upper bound in \cref{lemma:expval_variance} is very
pessimistic. Similar, constant factor looseness of the bounds was already observed  for conventional shadows~\cite{shadow_huang_2020},
however the discrepancy is strongly expected to be even larger for PEC shadows.
This is due to (b) as we do not take into account properties of the individual circuits in the quasi
probability decomposition but rather apply a pessimistic global bound.

\subsection{Sample complexity for predicting linear properties (\texorpdfstring{\cref{theo:m_properties}}))}
\label{app:proof_m_properties}
In order to predict expected values of $M$ independent observables $\lbrace O_{1},...,O_{M} \rbrace$, we group $\nshot= N_{batch} K$ independent snapshots into $K$ batches
$\mathcal{B}_{1},...,\mathcal{B}_{K}$ each of size $N_{batch}$. Then for each subset $\mathcal{B}_{i}$
one uses the empirical mean as $\hat{\mu}_{i}(O_{j}) = N_{batch}^{-1} \sum_{l \in \mathcal{B}_{i}} \text{Tr}[O_{j} (\hat{\rho}_{id})_{l}]$.
The final estimate for the expectation value of $O_{j}$ is then obtained by the median of the individual
empirical means, i.e., 
\begin{align}
	\hat{\mu}_{K,b}(O_{j}) := \text{median} \lbrace \hat{\mu}_{1}(O_{j}) ,..., \hat{\mu}_{K}(O_{j})  \rbrace.
\end{align}
Even though this method requires an increased number $N_{batch} K$ of independent classical shadows, it is much more robust against
outlier corruption. The idea is that if $\hat{\mu}_{K,b}(O_{j})$ deviates more than $\epsilon$ from
$\text{Tr}[O_{j} \rhoid]$, more than $K/2$ of the individual empirical mean values must deviate by more
than $\epsilon$. This is an exponentially suppressed event. This can be made more formal by the
concentration inequality estimator~\cite{mom_convergence_jerrum_1986,stat_learning_lerasle_2019}. 
\begin{align}
	\text{Prob}[\vert \hat{\mu}_{K,b}(O_{j}) - \langle O_{j} \rangle \vert \geq  \frac{2\sigma}{\sqrt{N_{batch}}}] \leq \exp ( -\frac{K}{8})
\end{align}
where $\sigma$ denotes the standard deviation. 
\begin{theorem}[Formal version of \cref{theo:m_properties}]\label{theo:m_prop_formal}
	Let $\mathcal{U}_{\text{circ}}$ be the ideal quantum circuit producing the ideal output state
	$\rho_{id}$ from \cref{def:decompose}.
	Suppose that we want to predict $M$ linear properties $O_{1}, \dots, O_{M}$ of the
	ideal state, i.e., $\langle O_{j} \rangle = \tr[O_{j} \rho_{id}]$.
	For fixed performance metrics $\epsilon, \delta \in [0,1]$ set
	\begin{align}
		N_{batch} = \frac{4g^{2}}{\epsilon^{2}} \, \underset{1 \leq j \leq M}{\lVert O_{j} \rVert_{E}^{2}} \quad \text{and} \quad K = 8 \log (\frac{M}{\delta})
	\end{align}
	Then a collection of $N =K N_{batch}$ independent classical shadows allow for accurately predicting
	all ideal expectation values via median of means estimation such that  
	\begin{align} \label{eq:concentration_mom}
		\mathrm{Prob}[\vert \hat{\mu}_{K,b}(O_{j}) - \langle O_{j} \rangle \vert \leq \epsilon] \geq 1 - \delta.
	\end{align}
\end{theorem}
\begin{proof}
	This is a direct consequence of the concentration property of the median of means estimator
	together with the bound on the variance from \cref{lemma:expval_variance}. 
	Because $\text{Var}[\hat{\mu}] \leq \lVert g \rVert_1^{2} \lVert O \rVert_{E}^{2}$ and if the accuracy
	is $\epsilon$,  we have
	$N_{batch} \geq 4 \lVert g \rVert_1^{2} \lVert O \rVert_{E}^{2} \geq  4 \sigma^{2} / \epsilon^{2}$.
	Further, as we have $M$ measurements that we want to accurately predict with at most
	failure probability $\delta$, we need for each individual measurement $\exp(- K/8) \leq \delta/M$.
	Thus the choice $K = 8 \log(M/\delta)$ yields the desired bound. In total we have:
	\begin{align*}
			&\text{Prob}[\vert \hat{\mu}_{K,b}(O_{j}) - \mu_{j} \vert \geq \epsilon \, \, \forall j ]  \\ 
            &= \text{Prob}[\bigcup_{j=1}^{M} \lbrace \vert \hat{\mu}_{K,b}(O_{j}) - \mu_{j} \vert 
			\geq \epsilon  \rbrace ] \\
			&\leq \sum_{j=1}^{M} \text{Prob}[\vert \hat{\mu}_{K,b}(O_{j}) - \mu_{j} \vert \geq \epsilon ] \leq M \frac{\delta}{M} = \delta
	\end{align*}
\end{proof}

\subsection{Predicting non-linear properties (\texorpdfstring{\cref{theo:nonlin}}))}
\label{app:nonlin}

In order to obtain rigorous
performance guarantees of our estimator, two ingredients are needed. First, note that any
polynomial function in the quantum state can be written as a linear function in tensor
products of the quantum state. More precisely, suppose we want to estimate a polynomial function of degree $m$ of the quantum state $\rho$, e.g., $\Tilde{f}: B(\mathcal{H}) \rightarrow \mathbb{R}$ with $\Tilde{f}(\rho) := \Tr[\Tilde{A} \rho^{m}]$, where $\rho, \Tilde{A} \in B(\mathcal{H})$. If $C^{(m)}: B(\mathcal{H}^{\otimes m}) \rightarrow B(\mathcal{H}^{\otimes m})$ denotes the cyclic permutation operator, i.e., $C^{(m)} (\vert \phi_{1} \rangle \vert \phi_{2} \rangle \cdots  \vert \phi_{m} \rangle) = \vert \phi_{m} \rangle \vert \phi_{1} \rangle  \cdots  \vert \phi_{m-1} \rangle$ we can associate to $\Tilde{f}$ a  function $f$ 
and an operator $A \in B(\mathcal{H}^{\otimes m})$ such that
\begin{align}\label{eq:predicting_nonlinearities}
    f(\rho) = \Tr[A \rho^{\otimes m}], \quad A = \Tr_{1}[C^{(m+1)} \Tilde{A} \otimes \openone^{\otimes m}]
\end{align}
and $f(\rho) = \Tilde{f}(\rho)$.

The second tool needed is the so called U-statistics, which often provides a uniformly
minimum variance unbiased estimator for nonlinear polynomial functions. Suppose we have
access to $N$ independent snapshots $\hat{\rho}_{1},...,\hat{\rho}_{N}$ which are generated by an
underlying state $\rho$ and that
$f(\hat{\rho}_{1},...,\hat{\rho}_{m})$ is a
polynomial function in the shadows such that $\theta$, which is our parameter of
interest is given by $\theta = \mathbb{E}[f(\hat{\rho}_{1},...,\hat{\rho}_{m})]$.
The $U$-statistics~\cite{u_statistics_hoeffding_1992} of order $m$  is defined as 
\begin{align}\label{eq:U_statistic}
	U_{N} := \binom{N}{m}^{-1} \sum_{\mathcal{C}_{N,m}} f(\hat{\rho}_{i_{1}},...,\hat{\rho}_{i_{m}}),
\end{align}
where $\mathcal{C}_{N,m}$ is the set of all combinations of $m$ distinct elements
that one can build out of $N$ different snapshots. The variance of this
estimator has a closed form in dependence on the function $f$.
For a $U$-statistic $U_{N}$ given by Eq.~\eqref{eq:U_statistic} the variance obeys~\cite{u_statistics_hoeffding_1992}  
\begin{align}\label{eq:var_summation}
	\text{Var}[U_{N}] = \frac{1}{\binom{N}{m}} \sum_{d=1}^{m} \binom{m}{d}
	\binom{N-m}{m-d} \text{Var}[f^{(d)}(\hat{\rho}_{1},...,\hat{\rho}_{d})],
\end{align}
where 
\begin{align*}
    f^{(d)} (\hat{\rho}_{1},...,\hat{\rho}_{d}) := \underset{\hat{\rho}_{d+1},...,\hat{\rho}_{m}}{\mathbb{E}} [f(\hat{\rho}_{1},...,\hat{\rho}_{d},\hat{\rho}_{d+1},...,\hat{\rho}_{m})].
\end{align*}
In order to understand the scaling of $\text{Var}[U_{N}]$ it is sufficient to consider a particular instance $\text{Var}[f^{(d)}(\hat{\rho}_{1},...,\hat{\rho}_{d})]$. First notice that for $A \in \mathcal{B}(\mathcal{H}^{\otimes m})$ as defined in Eq.~\eqref{eq:predicting_nonlinearities} one has 
\begin{align}
    f^{(d)} (\hat{\rho}_{1},...,\hat{\rho}_{k}) = \text{Tr}[A \hat{\rho}_{1}\otimes \cdots \otimes \hat{\rho}_{d} \otimes \rho^{\otimes p}]
\end{align}
where $p = m-d$ with the convention that $\rho^{\otimes 0} = 1 \in \mathbb{C}$.
Further define $\hat{\rho}_{\underline{l}} = \hat{\rho}_{l_{1}} \otimes \cdots \otimes \hat{\rho}_{l_{d}} \otimes \rho^{p}$
using the abbreviation $\hat{\rho}_{l_{1}} \equiv (\hat{\rho}_{id})_{l_{1}}$. For $1 \leq d \leq m$ we have 
\begin{align*}
    &\mathbb{E} [f^{(d)}(\hat{\rho}_{1},...,\hat{\rho}_{d})^{2}] \\ 
    &= \mathbb{E}[\text{Tr}(A \hat{\rho}_{1}\otimes \cdots \otimes \hat{\rho}_{k} \otimes \rho^{\otimes p})^{2}] \\
    &= \vert \vert g \vert \vert_{1}^{2d} \underset{\underline{k}_{1},l_{1}}{\mathbb{E}} \cdots \underset{\underline{k}_{d},l_{d}}{\mathbb{E}}  \text{Tr}[A \hat{\rho}_{\underline{l}}]^{2} \\
    &= \vert \vert g \vert \vert^{2d} \sum_{\underline{k}_{1},...,\underline{k}_{d}} \sum_{\underline{l}} (\prod_{j=1}^{d} p(\underline{k}_{j}) p(l_{j} \vert \underline{k}_{j})) \text{Tr}[A \hat{\rho}_{\underline{l}}]^{2} \\
    &=\vert \vert g \vert \vert^{2d} \sum_{\underline{k}_{1},...,\underline{k}_{d}} \sum_{\underline{l}} (\prod_{j=1}^{d} 
 \frac{   \vert g_{\underline{k}_{j}} \vert  } {\vert \vert g \vert \vert})
     \text{Tr}[\mathcal{G}_{\underline{k}_{j}} (\vert 0 \rangle \langle 0 \vert) E_{l_{j}})]  \text{Tr}[A \hat{\rho}_{\underline{l}}]^{2}
\end{align*}
Similarly as in the proof of \cref{lemma:expval_variance} we denote the operator
 $\Omega :=\vert \vert g \vert \vert^{-1} \sum_{\underline{k}_{j}} \vert g_{\underline{k}_{j}} \vert \mathcal{G}_{\underline{k}_{j}}$ for all $1 \leq j \leq d$ and write $E_{\underline{l}} = E_{l_{1}} \otimes \cdots \otimes E_{l_{d}}$. Then 
\begin{align}
\begin{split}
    &\vert \vert g \vert \vert^{2d} \sum_{\underline{l}} \prod_{j=1}^{d} \text{Tr}[\Omega(\vert 0 \rangle \langle 0 \vert  ) E_{\underline{l}_{j}}] \, \text{Tr}[A \hat{\rho}_{\underline{l}}]^{2} \\
    & =\vert \vert g \vert \vert^{2d} \sum_{\underline{l}} \text{Tr}[\Omega(\vert 0 \rangle \langle 0 \vert)^{\otimes d} E_{\underline{l}}] \, \text{Tr}[A \hat{\rho}_{\underline{l}}]^{2} \\
    & \leq \vert \vert g \vert \vert^{2d} \, \underset{\sigma}{\text{max}} \, \text{Tr}\Big[\Omega (\sigma)^{\otimes d} \sum_{\underline{l}} \text{Tr}[A \hat{\rho}_{\underline{l}}]^{2} E_{\underline{l}}\Big] \\
    &\leq \vert \vert g \vert \vert^{2d} \, \underset{\sigma}{\text{max}} \, \text{Tr}[\Omega(\sigma)^{\otimes d} \Gamma] \\
   &= \vert \vert g \vert \vert^{2d} \, \vert \vert \Gamma \vert \vert_{\infty},    
    \end{split}
\end{align}
where $\Gamma = \sum_{\underline{l}} \text{Tr}[A \hat{\rho}_{\underline{l}}]^{2} E_{\underline{l}}$.

Finally, we can evaluate an upper bound of the summation in \cref{eq:var_summation} analytically as
\begin{align*}
	\text{Var}[U_{N}] 
	\leq&
	 \vert \vert g \vert \vert^{2m} \, \vert \vert \Gamma \vert \vert_{\infty}  \frac{1}{\binom{N}{m}} \sum_{d=1}^{m} \binom{m}{d}
	\binom{N-m}{m-d} \\
	=&
	\vert \vert g \vert \vert^{2m} \, \vert \vert \Gamma \vert \vert_{\infty} 
	[ 1-\frac{((N-m)!)^2}{N! (N-2 m)!} ].
\end{align*}
Given the factorial formula is of $\mathcal{O}(1/N)$ with $N \equiv \nshot$,
the above bound implies that the number of samples needed to predict polynomial functions
of degree $m$ scales as $\mathcal{O}(\vert \vert g \vert \vert_{1}^{2m} / \epsilon^{2})$.

\section{Classical post-processing}
\subsection{Estimating local observables (\texorpdfstring{\cref{alg:reconstruct_algo}}))}
\label{app:reconstruct_algo}
For an arbitrary observable $P$ we can calculate the estimator from \cref{eq:mitigated_shadow_estimator} as
\begin{equation}
	\tr[ P \hat{\rho}_{\underline{k}, l} ] = \lVert g \rVert_1 \, \mathrm{sign}(g_{\underline{k}}) \tr[ P C_{E}^{-1} (E_{l}) ].
\end{equation}
In the idealised measurement case $E$ simplifies as detailed
in \cref{eq:idealised_measurements} and our classical shadow is
then a collection of $\nshot$ measurement outcomes $b \in \{ 0,1\}^N$ and
corresponding single-qubit Pauli measurement basis
defined by the single-qubit rotations $Q_k \in \mathcal{Q}$
Since $P$ is a $q$-local Pauli string it admits the product from
$P = \bigotimes_{i \in Q} P^{(i)} $ while the snapshot $C_{E}^{-1} (E_{l})$
similarly is of a product form via \cref{eq:snapshot}.
Note also that we use the index set $Q$ to abbreviate the set of qubits to which $P$
acts non-trivially and $|Q|=q$.
Thus we obtain the trace as
\begin{equation} \label{eq:prod}
	\tr[ P \hat{\rho}_{\underline{k}, l} ] 
	=
   \prod_{i \in Q} \tr\Big[ P^{(i)} \Big(3 (Q^{(i)}_{l})^{\dagger} \vert b^{(i)}\rangle \langle b^{(i)} \vert  Q^{(i)}_{l} - \openone \Big) \Big].
\end{equation}
Above we have used that the trace of a tensor product simplifies to a product of traces
and that on every qubit $i$ for which $P^{(i)} \equiv \openone$ the single-qubit expression evaluates to
\begin{equation*}
	\tr\Big[ P^{(i)} \Big(3 (Q^{(i)}_{l})^{\dagger} \vert b^{(i)}\rangle \langle b^{(i)} \vert  Q^{(i)}_{l} - \openone \Big) \Big] = 1
\end{equation*}
The expression in \cref{eq:prod} evaluates to $\{ \pm 3^q, 0 \}$ as we explain now.
The expression evaluates to $\pm 3^q$ if the measurement bases defined by $Q^{(i)}_k$
are the same as the single qubit Pauli matrices $P^{(i)}$ on the qubits $i \in Q$.
The sign is then determined by the bits $b^{(i)}$ in the bitstring $b$, i.e., it is negative
if the Hamming weight of the bitstring is odd on the qubits in $Q$.
Otherwise, if the measurement bases are not compatible with $P$ on the qubits in $Q$ then the above expression evaluates to zero.

Thus we obtain the simplified estimator as
\begin{equation*}
	\tr[ P \hat{\rho}_{\underline{k}, l}  ] = \lVert g \rVert_1 \, 3^q\, \mathrm{sign}(g_{\underline{k}}) f(b, Q_k),
\end{equation*}
where $f(b, Q_k) \in \{\pm 1, 0\} $.
The reconstruction algorithm thus takes the classical shadow data as the collection of bitsrings and 
Pauli measurement bases $\{b_k, P_k\}_{k=1}^{\nshot}$, as well as the Pauli observable $P$, and
calculates the values of $f(b,Q_k)$ as $0$ if the measurement bases are incompatible with $P$
and $\pm 1$ otherwise.
The algorithm has runtime $\mathcal{O}( q \nshot )$.

\subsection{Improved estimation of local observables with light cones (\texorpdfstring{\cref{alg:light_cone}}))}
\label{app:light_cone}

	For example, imagine a circuit of a single noisy gate with quasiprobaiblity decomposition $|\gamma_1| \mathcal{G}_1 - |\gamma_2| \mathcal{G}_2$
	and an observable $O$ whose light cone does not contain this gate.
	In \cref{alg:reconstruct_algo}
	we modify the sign $\sign(\gamma_k) \rightarrow +1$  and renormalise the coefficients such that
	$\lVert \gamma \rVert_1 \rightarrow 1$. Thus in expectation we implement the gate
	\begin{equation*}
		|\gamma_1| \tr[O \mathcal{G}_1(\rho_{in})] + |\gamma_2| \tr[O \mathcal{G}_2(\rho_{in})] = \tr[O \mathcal{U}(\rho_{in})],
	\end{equation*}
	where we used that $|\gamma_1| + |\gamma_2| = 1$ by definition.
	Furthermore, as the
	gate is not in the light cone of the observable $O$
	we used the identity $\tr[O \mathcal{G}_1(\rho_{in})] = \tr[O \mathcal{G}_2(\rho_{in})] = \tr[O \mathcal{U}(\rho_{in})]$
	where $\mathcal{U}$ is the ideal gate.
	
	In general, for general circuits with local noise models we can redefine the signs and norms in \cref{eq:product}   
	such that if the corresponding gate is not contained in the light cone as $l \notin \mathcal{I}$ then
	$\sign(\tilde{\gamma}^{(l)}_{k_l}) = +1$ with $\lVert \tilde{\gamma}^{(l)} \rVert_1 = 1$ 
	otherwise the coefficients are unchanged as
	$\tilde{\gamma}^{(l)}_{k_l}=  \gamma^{(l)}_{k_l}$.

	This allows us to rescale the original sampling cost $\lVert g\rVert_1 = \prod_{l=1}^\nu \lVert \gamma^{(l)} \rVert_1$
	where $\nu$ is the total number of noisy gates to $\prod_{l \in \mathcal{I}} \lVert \gamma^{(l)} \rVert_1$
	since we have redefined $ \lVert \gamma^{(l)} \rVert_1 \rightarrow 1$ for every gate whose index $l$ is outside of the light cone.		
	A significant advantage of this approach is that it is completely done at the stage of classical post-processing and
	at the time of measurement we simply just implement the conventional quasiprobability approach assuming all gates are noisy.
	For generalisation to non-local noise models refer to \cite{locality_error_mitigation_ibm_2023}.

\subsection{Local purities (\texorpdfstring{\cref{alg:entropy}}))}
For ease of notation we abbreviate the indexes of PEC snapshots as $\hat{\rho}_{i} := \hat{\rho}_{\underline{k}, l}$ with $i = (\underline{k}, l)$.
Given a subsystem as a set of qubits $Q =\left\{q_1, ..., q_m\right\}$
 	 we obtain an unbiased estimator for the Rényi entropy as
\begin{align}
	\hat{R}_Q &:= \tr\left[\text{SWAP}_{Q,Q'} \hat{\rho}_i  \otimes  \hat{\rho}_j \right] \nonumber \\
	& = \lVert g \rVert_1^2 \sign(g_{i})\sign(g_{j}) f( i, j, Q),
\end{align}
with $i \neq j$. Here $\text{SWAP}_{Q,Q'}$ swaps all pairs of qubits $q_k$ and $q_{k+N}$
in the system of $2N$ qubits in $\hat{\rho}_i  \otimes  \hat{\rho}_j $.

The factor $f( i, j, Q)$ can be computed analytically using that snapshots are of a product form as 
$\hat{\rho}_{i} = \bigotimes_{q=1}^N \hat{\rho}_{i}^{(q)}$ as
\begin{equation*}
f( i, j, Q)	=\prod_{q \in Q} \tr\left[ \text{SWAP}\hat{\rho}_{i}^{(q)}\otimes\hat{\rho}_{j}^{(q)}\right],
\end{equation*}
where $\text{SWAP}$ is the standard 2-qubit SWAP operator.
Here we have used that traces for qubits not in subsystem $Q$ evaluate to $1$ and that the trace
of a tensor product simplifies to a product of traces.

We can evaluate analytically the following  expression as it only involves 2 qubits as
$$\tr\left[ \text{SWAP}\hat{\rho}_{i}^{(q_k)}\otimes\hat{\rho}_{j}^{(q_k)}\right].$$
When the single-qubit measurement bases $Q_k^{(q)}$ and $Q_l^{(q)}$ in the snapshots are non-identical then the expression evaluates
to $\tfrac{1}{2}$.
When the measurement bases are identical then 
then the expression evaluates to $5$ given the measurement outcome bits are identical.
Otherwise it is $-4$ for non-identical measurement outcome bits. 
$f( i, j, Q)$ is then just a product of these values evaluated for all qubits in $Q$.

The algorithm simply iterates over all distinct pairs of snapshots and evaluates $f( i, j, Q)$.
We further multiply each snapshot outcome by the corresponding
signs $ \sign(g_{i})\sign(g_{j})$ and with the squared norm $\lVert g \rVert_1^2$. 
Finally, we compute the median of means of these individual outcomes.
The algorithm has a runtime $\mathcal{O}( |Q| \nshot^2 )$.

\section{Details of numerical simulations\label{app:numerics}}
In this work we use exact quantum state simulators to simulate noisy quantum circuits up to 12-qubit systems.
Given the present approach is effectively a Monte-Carlo sampling scheme, we efficiently simulate the effect of 
noise using a Monte-Carlo approach. For example, a Pauli channel is simulated by randomly
choosing a Pauli event according to its corresponding probability and applying the corresponding Pauli operator to the state.

As state-of-the-art experiments~\cite{kim2023evidence} apply Pauli twirling to guarantee that the noise model is 
well approximated by a Pauli channel---which can be learned efficiently~\cite{berg2022probabilistic}---we assume Pauli
noise models. In particular, we assume that two-qubit gates are the dominant source of noise and 
they are affected by Pauli errors with possibly
different probabilities for each gate as $p_k$.

\subsection{Details of \cref{fig:error_gs_estim}}

As we discussed in the main text, we simulate the 12-qubit ansatz circuit in \cref{fig:error_gs_estim} (left)
that prepares the ground state of the spin-ring Hamiltonian in \cref{eq:spin_ring_hamiltonian}.

\noindent \textbf{Gate error mitigation:}
We assume that each noisy quantum gate is affected by the following, biassed, local Pauli channel
\begin{align}
	\Phi_{k}(\rho) :=&  (1-p_k) \rho\\
	 +& p_k \left( \frac{1-\eta_k}{2} X \rho X + \frac{1-\eta_k}{2} Y \rho Y + \eta_k Z \rho Z \right),
	 \nonumber
\end{align} 
where $\eta_k$ is a bias parameter.
Furthermore, we consider that the error probability $p_k$ is specific to each gate and we randomly draw their values
from the distribution $\mathcal{N}(\mu = p, \sigma = p)$ and as we explain below
we explore two different regimes via $p \in \{10^{-3}, 2\times10^{-3} \}$.
We randomly choose probabilities to reflect the high
variability of two-qubit error rates in typical superconducting devices~\cite{kim2023evidence}. 
Of course, this variability may be lower in ion-trap devices~\cite{PRXQuantum.2.020343,moses2023race}.
Given in most platforms T2 relaxation timescales are significantly faster than T1 relaxation
timescales, we consider a biased noise channel, that is specific to each gate, by choosing the
bias parameters $\eta_k$ from the distribution $\mathcal{N}(\mu = 0.9, \sigma = 0.015)$.

We can straightforwardly find the inverse noise channel analytically as
\begin{equation*}
	\Phi_k(\rho)^{-1} :=  \gamma_0 \rho +  \gamma_1 X \rho X + \gamma_2 Y \rho Y + \gamma_3 Z \rho Z,
\end{equation*}
via coefficients 
\begin{align}
	\gamma_0 &= \lVert \gamma \rVert_1\frac{2 - p_k (p_k \eta_k^2 - p_k - 2 \eta_k + 4)}{2 + 2 p_k (\eta_k-1) (1 + p_k + p_k \eta_k)}, \nonumber \\
    \gamma_1 = \gamma_2 &= \lVert \gamma \rVert_1\frac{p_k (\eta_k-1) (p_k \eta_k + p_k -1)}{2 + 2 p_k (\eta_k-1) (1 + p_k + p_k \eta_k)}, \nonumber \\
    \gamma_3 &= \lVert \gamma \rVert_1\frac{p_k (2 \eta_k + p_k (\eta_k^2-1))}{2 + 2 p_k (\eta_k-1) (1 + p_k + p_k \eta_k)},
\end{align}
with $\lVert \gamma \rVert_1 = 1/2 \left(1/(1 + 2 p_k (\eta_k-1)) - 2/(p_k + p_k \eta_k-1)\right)$. 
In our simulations we thus randomly apply Pauli $X$, $Y$, $Z$ or $\openone$ gates
to the qubits in the support of the relevant gate via probabilities defined by $\underline{\gamma}$.

We additionally note that, while we found the coefficients $\underline{\gamma}$ analytically, they can also
conveniently be computed numerically at the pre-processing stage. The numerical inversion is efficient
given noise models are local, while
state-of-the art noise-model learning techniques achieve trivial inversion via
Pauli-Lindblad channels even for non-local (but sparse) models~\cite{berg2022probabilistic}.
Indeed, the present noise channel captures dominant terms in these experimentally-learned
error models~\cite{berg2022probabilistic}.

\noindent \textbf{Circuit error rate:}
In \cref{fig:error_gs_estim} we repeat our simulations for two different noise levels to demonstrate that the bias
(as well as the sampling overhead) grows exponentially.
We implement two different circuit error rates, $\xi \approx 0.15$ and $\xi \approx 0.26$.
As we detail in sec~\cref{sec:pec},
this circuit error rate $\xi = \sum_k p_k$ expresses the average number of errors in the full circuit 
and our circuit contains $\nu = 60$ noisy two-qubit gates.

\noindent \textbf{Readout error mitigation:}
We also consider readout errors: while readout errors may not be significant in ion-trap
devices~\cite{PRXQuantum.2.020343,moses2023race} as they are typically below gate operation
errors, we consider readout error rates that are consistent with typical superconducting systems.
In particular, we assume the noise model detailed in~\cref{sec:readout}
whereby the readout  of each qubit is affected by the same bitflip probability $\alpha = 0.01$.
We mitigate the effect of readout errors using the analytical
inverse of the measurement channel detailed in \cref{sec:readout}.

\noindent \textbf{Resampling scheme:}
	In \cref{fig:error_gs_estim} (middle) we perform a ground-state energy estimation for an increasing number $N_s$ of snapshots
	and for each fixed shot budget $N_s$, we estimate the average error from the exact energy by averaging over $10^4$ different experiments.
	We efficiently simulate this sampling task via the usual resampling scheme:
	we generate a very large pool of $10^7$ snapshots (an order of magnitude more than the largest
	simulated shot budget $N_s \leq 10^6$) using the above described
	quantum-mechanical simulations. We then estimate the ground-state energy by randomly
	choosing $N_s$ snapshots each time from this large pool.

\noindent \textbf{Error bounds:}
In \cref{fig:error_gs_estim} (right) we compare to analytical error bounds in \cref{theo:m_properties}.
Recall that our bounds in \cref{theo:m_properties} depend on the largest shadow norm $\max_{1 \leq k \leq M}{\lVert O_{k} \rVert_{E}^{2}}$
and in \cref{lemma:shadow_norm} we evaluate the shadow
norm of a $q$-local Pauli string as $3^q$ assuming ideal measurements via the snapshots of the from 
in \cref{eq:snapshot}.
We now take into account the effect of readout errors by explicitly calculating the shadow norm
under the above readout-error model.
As such, it is straightforward to modify our proof in \cref{lemma:shadow_norm}
by considering the effects and snapshots from \cref{sec:readout}, obtaining the following
bound on the shadow norms for considering all $q$ local
Pauli strings as
\begin{equation*}
	\max_{1 \leq k \leq M}{\lVert O_{k} \rVert_{E}^{2}} = 
	\frac{3^q}{(1{-}2\alpha)^{2q}}.
\end{equation*}
Indeed, readout-error mitigation has an associated measurement overhead of $(1{-}2\alpha)^{-2q}$.

\subsection{Details of \cref{fig:extrapol_shadows}}
As our main objective is to demonstrate that extrapolation-based error mitigation
is indeed compatible with classical shadows, we assume a simple error model whereby
every two-qubit gate has the same error probability $p$ that we can perfectly magnify
to $\lambda p$. Of course, in reality magnifying the error rates precisely is 
rather involved but has been successfully demonstrated in even large-scale experiments~\cite{kim2023evidence}.
In particular, in our simulations for extrapolation-based error mitigation, we assume
the above Pauli noise with an asymmetry parameter of $\eta_k = \tfrac{1}{3}$,
effectively a local depolarising noise with a probability $p$.
We define this channel as
\begin{equation}\label{eq:noise_channel}
	\Phi_p(\rho) :=  (1-p) \rho + p/3 [ X \rho X + Y \rho Y + Z \rho Z],
\end{equation} 
where $X$, $Y$ and $Z$ are Pauli matrices,
which channel can be analytically inverted to obtain the inverse channel as
\begin{equation*}
	\Phi_p(\rho)^{-1} :=  \gamma_0 \rho +  \gamma_1 X \rho X + \gamma_2 Y \rho Y + \gamma_3 Z \rho Z.
\end{equation*}
The explicit form of the coefficients follows as
\begin{equation*}
	\underline{ \gamma } =  \lVert \gamma \rVert_1 (
	\frac{3-p}{2 p +3},
	\frac{-p}{2 p+3},
	\frac{-p}{2 p+3},
	\frac{-p}{2 p+3}   ),
\end{equation*}
and the norm is $\lVert \gamma \rVert_1 = (3 + 2 p)/(3 - 4 p)$.
The same noise model was also used for \cref{fig:entropy}.


%

\end{document}